\documentclass[letterpaper,11pt,reqno]{amsart}

\makeatletter
\usepackage{amssymb}
\usepackage{latexsym}
\usepackage{amsbsy}
\usepackage{amsfonts}
\usepackage{caption}
\usepackage{subcaption}
\usepackage{hyperref}
\usepackage{graphicx}
\usepackage{enumerate}
\usepackage{enumitem}
\usepackage{color}
\usepackage[round]{natbib} 
\usepackage{comment}
\usepackage{scalerel}
\usepackage{tikz}

\def\marginpar#1{\ignorespaces}

\DeclareMathOperator\bet{Beta}
\DeclareMathOperator\Dir{Dir}
\DeclareMathOperator\var{Var}
\DeclareMathOperator\GEM{GEM}
\newtheorem{theorem}{Theorem}[section]
\newtheorem{lemma}[theorem]{Lemma}

\newtheorem{definition}[theorem]{Definition}

\numberwithin{equation}{section}
\makeatother
\begin{document}
\title[Proof of Stake]{Stability of shares in the Proof of Stake Protocol \\ -- Concentration and Phase Transitions}

\author[Wenpin Tang]{{Wenpin} Tang}
\address{Department of Industrial Engineer and Operations Research, Columbia University. Email: 
} \email{wt2319@columbia.edu}

\date{\today} 
\begin{abstract}
This paper is concerned with the stability of shares in a cryptocurrency where the new coins are issued 
according to the Proof of Stake protocol. 
We identify large, medium and small investors under various rewarding schemes, 
and show that the limiting behaviors of these investors are different --
for large investors their shares are stable, 
while for medium to small investors their shares may be volatile or even shrink to zero. 
For instance, with a geometric reward there is chaotic centralization,
where all the shares will eventually concentrate on one investor in a random manner.
This leads to the phase transition phenomenon, and the thresholds for stability are characterized.
In response to the increasing activities in blockchain networks,
we also propose and analyze a dynamical population model for the PoS protocol, which allows the number of investors to grow over the time.
Numerical experiments are provided to corroborate our theory.
\end{abstract}
\maketitle

\textit{Key words}: Blackwell-MacQueen urn, blockchain, concentration/anti-concentration, 
cryptocurrency, P\'olya urn, Proof of Stake.

\maketitle

\setcounter{tocdepth}{1}
\section{Introduction}
\label{sc1}

\quad In the past decade, blockchain technology has evolved tremendously, and is now regarded as an endeavor 
to facilitate the next generation digital exchange platform with a broad range of established or emerging applications including
cryptocurrency (\cite{Naka08, Wood14, HR17}), 
healthcare (\cite{MC19, Do19, TPE20}),
supply chain (\cite{SK19, CTT20}), 
electoral voting (\cite{W18}) and new-born arts as non-fungible tokens (\cite{WL21, Dow22}).

\quad A blockchain is a growing chain of records or transactions, called {\em blocks}, 
which are jointly maintained by a set of {\em miners} or {\em validators} using cryptography.
The idea of blockchain originated from {\em distributed consensus}, where multiple machines in a mission-critical system are required to make consistent decisions. 
The work of \cite{LSP82}, which introduced the famous {\em Byzantine Generals problem}, laid the foundation for distributed consensus. 
From then on, distributed consensus has been deployed in many digital infrastructures such as Google Wallet and Facebook Credit. 
Since 2009, Bitcoin (\cite{Naka08}) and various other cryptocurrencies have come around,
and allowed the secure transfer of assets without an intermediary such as a bank or payment processing network. 
These cryptocurrencies achieved a new breakthrough in distributed consensus
because they showed for the first time that consensus is viable in a decentralized and permissionless environment 
where {\em anyone} is allowed to work. 
This is in contrast with traditional trusted payment processing systems
as well as distributed consensus-based computing infrastructures such as Google Wallet and Facebook Credit, 
where only a small number of permissioned personnels can participate.
This way, Google Wallet and Facebook Credit are regarded as {\em permissioned blockchains},
while Bitcoin and other cryptocurrencies are {\em permissionless blockchains}.

\quad At the core of a blockchain is the {\em consensus protocol} or {\em smart contract},
which specifies a set of voting rules allowing miners to agree on an ever-growing, linearly-ordered log of transactions forming the distributed ledger.
There are several existing blockchain protocols, 
among which the most popular are {\em Proof of Work} (PoW, \cite{Naka08}) and {\em Proof of Stake} (PoS, \cite{KN12, Wood14}). 
\begin{itemize}[itemsep = 3 pt]
\item
In the PoW protocol, miners compete with each other by solving a puzzle, 
and the one who solves the puzzle first is allowed to append a new block to the blockchain.
Thus, the probability of a miner being selected is proportional to the computational power that the miner has.
The PoW coins include Bitcoin, Ethereum, Dogecoin...etc.
\item
In the PoS protocol, the blockchain is updated by a randomly selected miner
where the probability of a miner being drawn is proportional to the stake that the miner owns.
The PoS coins include BNB, Cardano, Solana...etc.
\end{itemize}

\quad As of May 28, 2022, Cryptoslate lists $352$ PoW coins with a total $\$ 800$B ($66 \%$) market capitalization, 
and $279$ PoS coins contributing a $\$ 122$B ($10\%$) market capitalization. 
One disadvantage of the PoW protocol is that competition among miners has led to exploding levels of energy consumption.
For instance, \cite{Mora18} pointed out the unsustainability of PoW-based blockchains,
and \cite{PS21} showed that energy consumption of Bitcoin is at least $10,000$ times higher than PoS-based blockchains. 
\cite{CK17, Saleh19, CHL20} also discussed drawbacks of PoW blockchains from economic perspectives. 
Though PoS-based blockchains are still not as widely used as PoW-based ones, 
there is a strong incentive among blockchain practitioners to switch from a PoW to PoS ecosystem
as is the case for Ethereum $2.0$ (\cite{Wick21}).

\quad A PoS blockchain mitigates the problem of energy efficiency and scalability; however, 
it receives several criticisms. 
\begin{itemize}[itemsep = 3 pt]
\item
{\em Security concern}. 
PoS protocols may suffer from the {\em Nothing-at-Stake} problem, 
where it is effortless for (malicious) miners to copy every outdated history and participate in all of them.
As shown in \cite{BD19}, any miner with more than $1/(1+e) \approx 27 \%$ of the total stake can launch a {\em double-spending attack},
which seems to be less robust than the $51\%$ attack to a PoW-based blockchain.
\cite{BN19} also pointed out difficulties in detecting Nothing-at-Stake for incentive-compatible PoS protocols.
One way to address this challenge is to come up with a clever block rewarding scheme,
as discussed in \cite{KR17, DPS19, Saleh21}.
\item
{\em Centralization concern}.
Critics have argued that the PoS protocol induces wealth concentration, thus leading to centralization (\cite{AF18, IR18, FK19}).
As one key feature of a permissionless blockchain is decentralization, 
the problem of centralization may put questions to the necessity of using a permissionless blockchain
since it yields concentration in voting mechanism as does a permissioned blockchain.
Previous works (\cite{AC20, AW22}) showed that the PoW protocol may generate concentration in mining power, and thus centralization in various settings. 
At a conceptual level, the effect of centralization or decentralization in both PoW and PoS blockchains arises from randomness. 
In a PoW blockchain, the randomness is external to the cryptocurrency,
while in a PoS blockchain, the randomness comes from the cryptocurrency itself.
\end{itemize}

\quad In this paper we consider the problem of wealth stability
for different types of miners of a PoS blockchain.
In the sequel, a PoS blockchain miner is also called an {\em investor}.
The prominent work of \cite{RS21} took the first step to study the long term evolution of large investor
shares in a PoS cryptocurrency via a {\em P\'olya urn model}. 
They proved under various reward assumptions that the shares in the long run do not deviate much from the initial ones 
as the initial coin offerings are large;
that is,
\begin{equation*}
[\mbox{share at time } \infty] - [\mbox{share at time } 0] \approx 0 \quad \mbox{for each investor}.
\end{equation*}
That is to say, the {\em rich-get-richer phenomenon} does not occur when there is a small number of investors each having a large proportion of initial coins. 
Statistical analysis from (\cite{FKP19, TGT20}) suggested that there be increasing levels of activities from smaller investors in blockchain networks.
Moreover, online platforms such as Robinhood Crypto allow for fractional trading. 
Thus, it is indispensable to understand the evolution of small investor shares as well.
Typically for these investors, 
$[\mbox{share at time 0}] \approx 0$
so the above relation holds trivially as `$0 - 0 = 0$'.
This prompts us to consider a more meaningful measure of evolution --
the {\em ratio} $[\mbox{share at time } \infty] / [\mbox{share at time } 0]$,
where `$0/0$' is indeterminate for small investors.

\quad The first contribution of this work is to provide a systematic study of the aforementioned ratio
under various rewarding schemes and for different types of investors
in the setting of \cite{RS21}.
The investors are categorized into {\em large}, {\em medium} and {\em small} ones in terms of their initial endowment of coins.
The key observation is that for investors whose initial coins scale differently to the total initial coin offering, 
the ratio exhibits different asymptotic behaviors such as concentration and anti-concentration.
This is a {\em phase transition} phenomenon, which has yet appeared in the literature on the economics of the PoS protocol.
For instance, we prove that under a constant rewarding scheme, 
the ratio of evolution of shares concentrates at one for large investors, 
and converges to a Gamma random variable for medium investors,
and decays towards zero for small investors (Theorem \ref{thm:1}).
Similar but slightly weaker results are established under a decreasing rewarding scheme (Theorem \ref{thm:2})
and an increasing rewarding scheme (Theorem \ref{thm:3}).
In the case of a geometric reward, our result shows that with probability one, 
all the shares will eventually go to one investor in a completely random way.
Such a behavior, which we call {\em chaotic centralization}, indicates that 
the geometric rewarding scheme induces long-run instability in a PoS blockchain.
This, however, is not contradictory to \cite{FK19} where it was shown that a geometric reward is optimal in a {\em fixed} time horizon. 

\quad As is clear from the growing popularity of blockchain technology, it may not be realistic to assume a {\em fixed} finite number of investors throughout.
The second contribution of this paper is to propose a dynamical approach, which starts with a finite number of investors and evolves to an infinite number of investors.
Our model can be viewed as an infinite population version of the P\'olya urn model, 
and the idea comes from the {\em Blackwell-MacQueen urn model} in Bayesian nonparametrics, as well as {\em species sampling} in computational biology.
In this dynamical population setting, 
we study the ratio of evolution of shares under various rewarding schemes (Theorem \ref{thm:6}). 
Our result shows that under a decreasing rewarding scheme with a suitably large decay,
the shares of the initial capitalists will be diluted in the long run. 
This observation is consistent with several existing practice (e.g. Bitcoin),
and it may also provide guidance on the choice of rewarding schemes 
so that decentralization is incorporated in a PoS blockchain. 

\quad Let us also mention some relevant works.
\cite{TF20} studied the stake-based voting for crowd-sourcing on a blockchain,
and showed that asymmetric information may lead to suboptimal outcomes. 
\cite{BFT21} considered a committee-based protocol, and highlighted its advantage over other PoS protocols.
 To the best of our knowledge, it is the first time in this work that heterogeneity of investors is taken into account, 
leading to a phase transition in centralization vs decentralization.
This paper also proposes and studies for the first time a dynamical model in which the number of investors is not fixed 
for the PoS protocol.
Our paper thus contributes to both the literature on the decentralization of blockchains
and that on the economics of the PoS protocol.

\quad The remainder of the paper is organized as follows.
In Section \ref{sc2}, we adopt the P\'olya urn to model the PoS protocol, and study the ratio of evolution of shares under various rewarding schemes.
In Section \ref{sc3}, we propose a dynamical population model
which allows the number of investors to increase over the time.
There we also provide analysis on the ratio of evolution of shares. 
We conclude with Section \ref{sc4}.
All the proofs are given in Appendix.

\section{Finite population model}
\label{sc2}

\quad In this section, we follow \cite{RS21} to use a time-dependent P\'olya urn model to describe the PoS protocol. 
Below we collect the notations that will be used throughout this paper.
\begin{itemize}[itemsep = 3 pt]
\item
$\mathbb{N}_{+}$ denotes the set of positive integers, and $\mathbb{R}_{+}$ denotes the set of positive real numbers.
\item
The symbol $a = \mathcal{O}(b)$ means that $a/b$ is bounded from above as $b \to \infty$;
the symbol $a = \Theta(b)$ means that $a/b$ is bounded from below and above as $b \to \infty$;
and the symbol $a = o(b)$ or $b \gg a$ means that $a/b$ decays towards zero as $b \to \infty$.
\end{itemize}

\quad Let $K \in \mathbb{N}_{+}$ be the number of investors, 
and $N \in \mathbb{R}_{+}$ be the number of initial coins/tokens in a PoS blockchain.
The investors are indexed by $[K]: = \{1, \ldots, K\}$, and investor $k$'s initial endowment of coins is $n_{k,0}$ with $\sum_{k = 1}^K n_{k,0} = N$.
We define the {\em investor share} as the fraction of coins each investor owns. 
So the initial investor shares $(\pi_{k, 0}, \, k \in [K])$ are given by
\begin{equation}
\label{eq:share0}
\pi_{k, 0}: = \frac{n_{k,0}}{N}, \quad k \in [K].
\end{equation}
Similarly, we denote by $n_{k,t}$ the number of coins owned by investor $k$ at time $t \in \mathbb{N}_{+}$, and the corresponding share is
\begin{equation}
\label{eq:sharet}
\pi_{k,t}:= \frac{n_{k,t}}{N_t}, \quad k \in [K], \quad \mbox{with } N_t:= \sum_{k=1}^K n_{k,t}.
\end{equation}
Here $N_t$ is the total number of coins at time $t$, and thus $N_0 = N$.
Clearly, for each $t \ge 0$ $(\pi_{k, t},\, k \in [K])$ forms a probability distribution on $[K]$.

\quad Now we provide a formal description of the PoS dynamics. 
At time $t \in \mathbb{N}_{+}$, investor $k$ is selected at random among $K$ investors with probability $\pi_{k,t-1}$.
Once selected, the investor receives a deterministic reward of $R_t \in \mathbb{R}_{+}$ coins.
Let $S_{k,t}$ be the random event that investor $k$ is selected at time $t$. 
Thus, the number of coins owned by each investor evolves as  
\begin{equation}
\label{eq:TPolya}
n_{k,t} = n_{k-1, t} + R_t 1_{S_{k,t}}, \quad k \in [K].
\end{equation}
It was shown in \cite[Proposition 5]{RS21} that investors have no incentive to trade their coins under some risk neural conditions. 
Without loss of generality, we exclude the possibility of exchanges among investors. 
Note that the total number of coins satisfies $N_t = N_{t-1} + R_t$. 
Combining \eqref{eq:sharet} and \eqref{eq:TPolya} yields a recursion of the investor shares:
\begin{equation}
\label{eq:Dshare}
\pi_{k,t} = \frac{N_{t-1}}{N_t} \pi_{k, t-1} + \frac{R_t}{N_t} 1_{S_{k,t}}, \quad k \in [K].
\end{equation}
Let $\mathcal{F}_t$ be the filtration as the $\sigma$-algebra generated by the random events $(S_{k, r}: k \in [K], r \le t)$.
It is easily seen that for each $k \in [K]$, the process of investor $k$'s share $(\pi_{k,t}, \, t \ge 0)$ is an $\mathcal{F}_t$-martingale.
By the martingale convergence theorem (see e.g. \cite{Durrett}), 
\begin{equation}
(\pi_{1,t}, \ldots, \pi_{K,t}) \longrightarrow (\pi_{1,\infty}, \ldots, \pi_{K,\infty}) \quad \mbox{as } t \rightarrow \infty \mbox{ with probability }1,
\end{equation}
where $(\pi_{1,\infty}, \ldots, \pi_{K,\infty})$ is some random probability distribution on $[K]$.

\quad We consider the evolution of the investor shares $(\pi_{k,t}, \, k \in [K])$, as well as its long time limit $(\pi_{k,\infty}, \, k \in [K])$.
One major problem is to know whether the PoS protocol triggers the rich-get-richer, or concentration phenomenon by comparing the investor shares at the initial time and in the long time limit.
As briefly discussed in the introduction, 
\cite{RS21} were concerned with large investors, i.e. $\pi_{k, 0} = \Theta(1)$ or equivalently $n_{k,0} = \Theta(N)$.
They proved under various reward assumptions that 
\begin{equation}
\label{eq:RSres}
\lim_{N \to \infty} \mathbb{P}(|\pi_{k, \infty} - \pi_{k,0}| > \varepsilon) = 0 \quad \mbox{for each } k \in [K] \mbox{ and each fixed } \varepsilon > 0.
\end{equation}
That is, the rich-get-richer phenomenon does not occur typically when there is a small number $K = \Theta(1)$ of investors each having a significant proportion $\pi_{k, 0} = \Theta(1)$ of initial coins. 
However, there may also be medium or even small investors whose initial shares $\pi_{k,0} = o(1)$ is relatively small.
For instance, 
\begin{itemize}[itemsep = 3 pt]
\item
When the number of investors $K \approx N$, there are less rich investors with initial endowment of coins $n_{k, 0} = f(N)$ such that $f(N) \to \infty$ but $f(N)/N \to 0$, and small investors with initial number of coins $n_{k,0} = \Theta(1)$.
\item
When the number of investors $K \gg N$, there may also exist small investors who own fractional number of coins $n_{k,0} = o(1)$.
\end{itemize}
In these cases, we have $\pi_{k,0} = o(1)$,
so for each fixed $\varepsilon > 0$ the probability $\mathbb{P}(|\pi_{k, \infty} - \pi_{k,0}| > \varepsilon)$ always goes to $0$ as $N \to \infty$.
Thus, it is more reasonable to consider the ratios $\pi_{k, t}/\pi_{k,0}$ or $\pi_{k, \infty}/\pi_{k,0}$ instead of the differences.  
In the following subsections,
we study the ratio of evolution of shares under various rewarding schemes and for different types of investors, 
encompassing all the above instances.
In particular, we observe phase transitions for different types of investors in terms of wealth stability.

\subsection{Constant reward}
We start with the constant reward $R_t \equiv R > 0$ (e.g. Blackcoin).
Let $\Gamma(z):=\int_0^{\infty} x^{z-1} e^{-x} dx$ be the Gamma function.
Recall that the Dirichlet distribution with parameters $(a_1, \ldots, a_K)$, which we simply denote $\Dir(a_1, \ldots, a_K)$, has support on the standard simplex $\{(x_1, \ldots, x_K) \in \mathbb{R}_{+}^K: \sum_{k = 1}^K x_k = 1\}$ and has density
\begin{equation}
\label{eq:Dirichlet}
f(x_1, \ldots, x_K) = \frac{\Gamma\left(\sum_{k=1}^K a_k \right)}{\prod_{k=1}^K \Gamma(a_k)} \prod_{k=1}^K x_k^{a_k-1}.
\end{equation}
For $K=2$, the Dirichlet distribution reduces to the beta distribution which is denoted by $\bet(a_1, a_2)$.
It is easily seen that if $(x_1, \ldots,x_K) \stackrel{d}{=} \Dir(a_1, \ldots, a_K)$ then for each $k \in [K]$
$x_k \stackrel{d}{=} \bet(a_k, \sum_{j \ne k} a_j)$.
The following result is concerned with the evolution of shares in a PoS protocol with a constant reward.
\begin{theorem}
\label{thm:1}
Assume that the coin reward is $R_t \equiv R > 0$.
Then the investor shares have a limiting distribution
\begin{equation}
\label{eq:limDirichlet}
(\pi_{1,\infty}, \ldots, \pi_{K,\infty}) \stackrel{d}{=} \Dir\left(\frac{n_{1,0}}{R}, \ldots,\frac{n_{K,0}}{R}\right).
\end{equation}
Moreover, 
\begin{enumerate}[itemsep = 3 pt]
\item[(i)]
For $n_{k,0} = f(N)$ such that $f(N) \to \infty$ as $N \to \infty$ (i.e. $\pi_{k,0} \gg 1/N$), we have for each $\varepsilon > 0$ and each $t \ge 1$ or $t = \infty$:
\begin{equation}
\label{eq:coninter}
\mathbb{P}\left(\left|\frac{\pi_{k, t}}{\pi_{k,0}} - 1\right| > \varepsilon \right) \le \frac{5 R}{4 \varepsilon^2 f(N)},
\end{equation}
which converges to $0$ as $N \to \infty$.
\item[(ii)]
For $n_{k,0} = \Theta(1)$ (i.e. $\pi_{k,0} = \Theta(1/N)$), we have the convergence in distribution:
\begin{equation}
\label{eq:consmall}
\frac{\pi_{k,\infty}}{\pi_{k,0}} \stackrel{d}{\longrightarrow} \frac{R}{n_{k,0}} \gamma\left( \frac{n_{k,0}}{R}\right) \quad \mbox{as } N \to \infty,
\end{equation}
where $\gamma\left(\frac{n_{k,0}}{R}\right)$ is a Gamma random variable with density $x^{\frac{n_{k,0}}{R}-1} e^{-x} 1_{x > 0}/\Gamma\left(\frac{n_{k,0}}{R}\right)$.
\item[(iii)]
For $n_{k,0} = o(1)$ (i.e. $\pi_{k,0} = o(1/N)$), we have $\var(\pi_{k,\infty}/\pi_{k,0}) \to \infty$ as $N \to \infty$.
Moreover, for each $\varepsilon > 0$:
\begin{equation}
\label{eq:continy}
\mathbb{P}\left(\frac{\pi_{k,\infty}}{\pi_{k,0}} < \varepsilon \right) \to 1 \quad \mbox{as } N \to \infty.
\end{equation}
\end{enumerate}
\end{theorem}

\quad The proof of Theorem \ref{thm:1} is given in Appendix \ref{scA}.
Let us make a few comments. 
The theorem reveals a phase transition of shares in the long run between large, medium and small investors. 
Part ($i$) shows that for large investors, their shares are stable in the sense that the ratio between the share at a later time $t$ and the initial share, i.e. $\pi_{k,t}/\pi_{k,0}$, converges in probability to $1$ as the initial coin offerings $N \to \infty$.
In particular, for extremely large investors with initial coins $n_{k,0} = \Theta(N)$ 
this is equivalent to the stability as defined by \eqref{eq:RSres}.
The stability also holds for less rich large investors with $n_{k,0} \gg 1$ but $n_{k,0} = o(N)$.
Note that the concentration bound \eqref{eq:coninter} is uniform in time $t$, implying that the rich-get-richer phenomenon does not occur at any time.
See Figure \ref{fig:1} for an illustration of the bound.

\begin{figure}[htb]
\includegraphics[width=0.45\columnwidth]{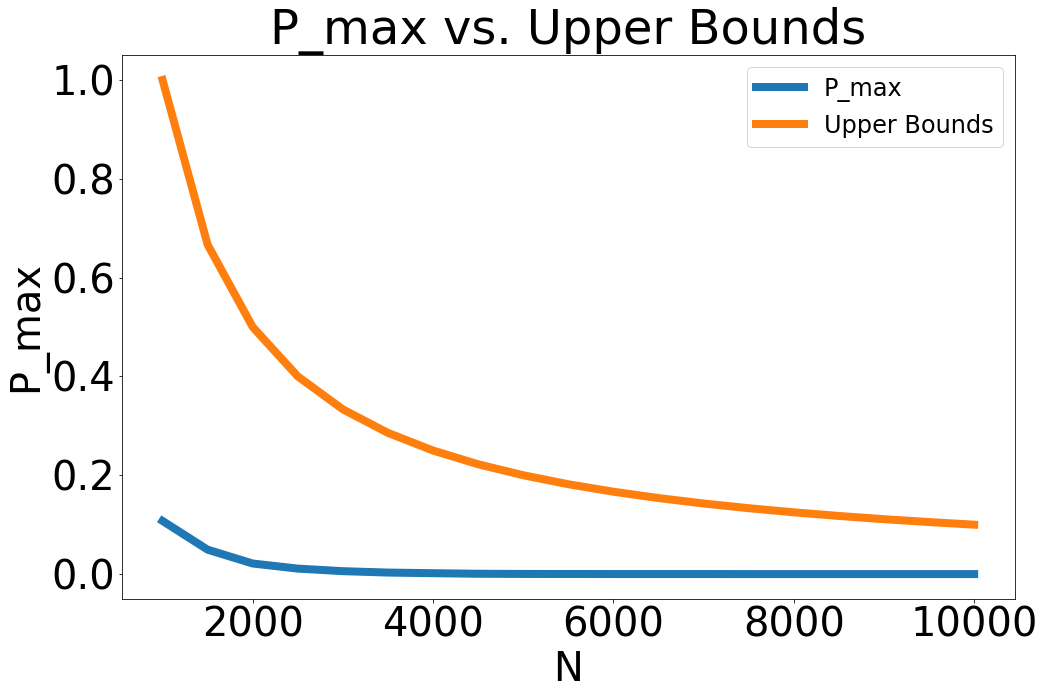}
\caption{Constant reward: stability of $\pi_{k,t}/\pi_{k,0}$ for large investors.
Blue curve: $P_{\max}$ is a MC estimate of $\max_{1 \le t \le 50000} \mathbb{P}\left(\left| \frac{\pi_{k,t}}{\pi_{k,0}}\right| > 0.05 \right)$.
Orange curve: right side upper bound in \eqref{eq:coninter} with $R = 1$, $\varepsilon = 0.05$, $n_{k,0} = N/2$ and $N \in \{1000, 1500, 2000, \ldots, 10000\}$.}
\label{fig:1}
\end{figure}

\quad On the other hand, the evolution of shares for medium or small investors has very different limiting behaviors. 
Part ($ii$) shows that medium investor's shares are volatile in such a way that the ratio $\pi_{k,\infty}/\pi_{k,0}$ is approximated by a gamma distribution independent of the initial coin offerings, and hence $\var(\pi_{k,\infty}/\pi_{k,0}) \approx \frac{1}{n_{k,0}}$.
For instance, if $n_{k,0} = R = 1$ the limiting distribution of the ratio $\pi_{k,\infty}/\pi_{k,0}$ reduces to the exponential distribution with parameter $1$. 
See Figure \ref{fig:2A} for an illustration of this approximation.
In this case, we have
\begin{equation*}
\mathbb{P}\left(\frac{\pi_{k,\infty}}{\pi_{k,0}} >  \theta \right) \approx e^{-\theta} \quad \mbox{as } N \to \infty.
\end{equation*}
So with probability $e^{-2} \approx 0.135$ a medium investor's share will double, and with probability $1 - e^{-0.5} \approx 0.393$ this investor's share will be halved. 
Part ($iii$) shows that for small investors, their shares will be shrinking along the time, and the ratio $\pi_{k,\infty}/\pi_{k,0}$ converges to $0$ in probability as $N \to \infty$.
This is indeed the {\em poor-get-poorer} phenomenon as illustrated in Figure \ref{fig:2B}.

\begin{figure}[htb]
    \centering
\begin{subfigure}{0.45\textwidth}
  \includegraphics[width=\linewidth]{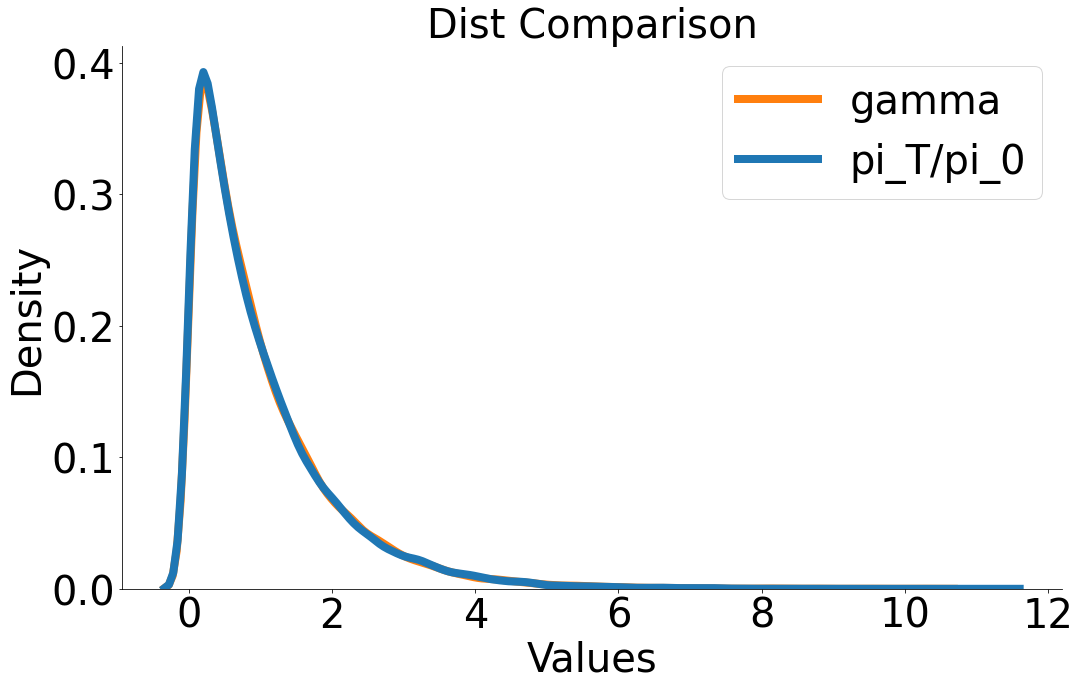}
  \caption{Blue curve: histogram of $\pi_{k,50000}/\pi_{k,0}$ with $n_{k,0} = R = 1$ and $N = 100$.
  Orange curve: Gamma distribution.}
  \label{fig:2A}
\end{subfigure}\hfil
\begin{subfigure}{0.45\textwidth}
  \includegraphics[width=\linewidth]{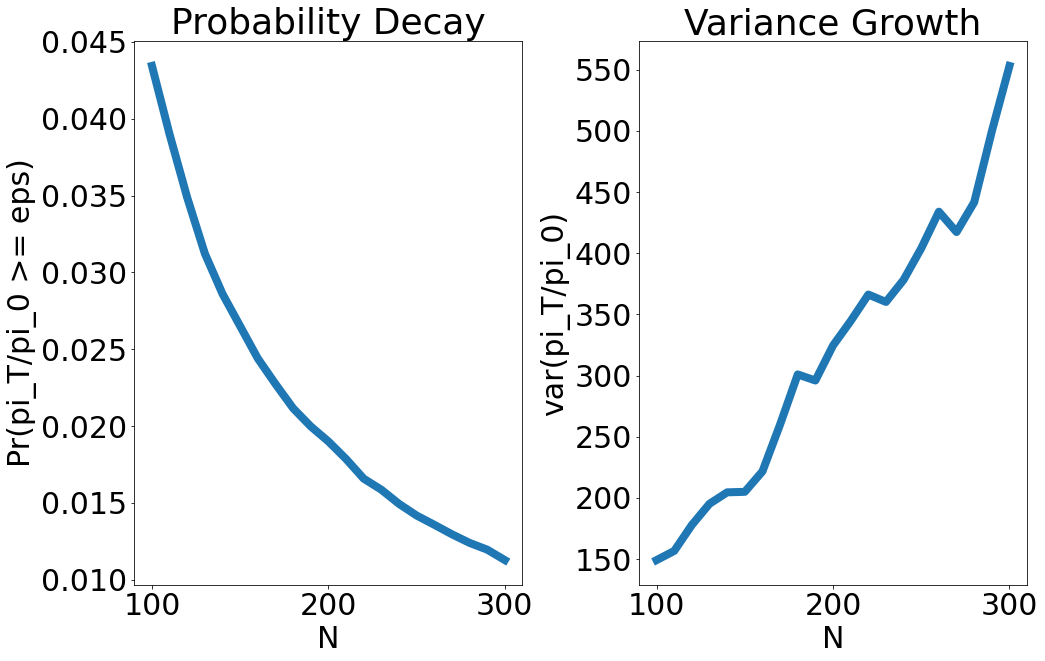}
  \caption{MC estimates of $\mathbb{P}\left(\pi_{k,50000}/\pi_{k,0} > 0.05 \right)$ and $\var\left(\pi_{k,50000}/\pi_{k,0} \right)$ with $R = 1$, $n_{k,0} = N^{-1.1}$ and $N \in \{100, 110, \ldots, 300\}$}
  \label{fig:2B}
\end{subfigure}\hfil
\caption{Constant reward: instability of $\pi_{k,t}/\pi_{k,0}$ for medium and small investors.}
\end{figure}

\quad Finally, we note that the results in Theorem \ref{thm:1} hold jointly for a finite number of investors belonging to the same category. 
That is,
\begin{itemize}[itemsep = 3 pt]
\item
For $n_{k_1, 0}, \ldots, n_{k_\ell, 0}$ with $k_1, \ldots, k_\ell \in [K]$ satisfying the conditions in ($i$),
\begin{equation*}
\left(\frac{\pi_{k_1,t}}{\pi_{k_1,0}}, \ldots, \frac{\pi_{k_\ell,t}}{\pi_{k_\ell,0}} \right) \rightarrow 1 \quad \mbox{in probability as } N \to \infty.
\end{equation*}
\item
For $n_{k_1, 0}, \ldots, n_{k_\ell, 0}$ with $k_1, \ldots, k_\ell \in [K]$ satisfying the condition in ($ii$),
\begin{equation*}
\left(\frac{\pi_{k_1,\infty}}{\pi_{k_1,0}}, \ldots, \frac{\pi_{k_\ell,\infty}}{\pi_{k_\ell,0}} \right) \stackrel{d}{\longrightarrow} \left(\frac{R}{n_{k_1,0}} \gamma\left(\frac{n_{k_1,0}}{R}\right), \ldots, \frac{R}{n_{k_\ell,0}} \gamma\left(\frac{n_{k_{\ell},0}}{R}\right) \right) \quad \mbox{as } N \to \infty,
\end{equation*}
where $\gamma(n_{k_1,0}/R), \ldots, \gamma(n_{k_{\ell},0}/R)$ are independent Gamma random variables. 
\item
For $n_{k_1, 0}, \ldots, n_{k_\ell, 0}$ with $k_1, \ldots, k_\ell \in [K]$ satisfying the conditions in ($iii$),
\begin{equation*}
\left(\frac{\pi_{k_1,t}}{\pi_{k_1,0}}, \ldots, \frac{\pi_{k_\ell,t}}{\pi_{k_\ell,0}} \right) \rightarrow 0 \quad \mbox{in probability as } N \to \infty.
\end{equation*}
\end{itemize}

\subsection{Decreasing reward}

We consider the decreasing reward such that $R_t \ge R_{t+1}$ for each $t \ge 0$ (e.g. Bitcoin).
The following theorem shows that the ratio of evolution of shares is more complicated 
in the PoS protocol with a decreasing reward scheme, and the phase transition may depend on how the reward function decreases over the time.

\begin{theorem}
\label{thm:2}
Assume that the coin reward is $R_t$ with $R_t \ge R_{t+1}$ for each $t \ge 0$.
\begin{enumerate}[itemsep = 3 pt]
\item
If $R_t$ is bounded away from $0$, i.e. $\lim_{t \ge 0} R_t = \underline{R} > 0$, then
\begin{enumerate}[itemsep = 3 pt]
\item[(i)]
For $n_{k,0} = f(N)$ such that $f(N) \to \infty$ as $N \to \infty$ (i.e. $\pi_{k,0} \gg 1/N$), we have for each $\varepsilon > 0$ and each $t \ge 1$ or $t = \infty$:
\begin{equation}
\label{eq:decinter1}
\mathbb{P}\left(\left|\frac{\pi_{k, t}}{\pi_{k,0}} - 1\right| > \varepsilon \right) \le \frac{R_1}{\varepsilon^2 f(N)},
\end{equation}
which converges to $0$ as $N \to \infty$.
\item[(ii)]
For $n_{k,0} = \Theta(1)$ (i.e. $\pi_{k,0} = \Theta(1/N)$), we have $\var \left(\frac{\pi_{k, \infty}}{\pi_{k,0}}\right) = \Theta(1)$.
Moreover, there is $c > 0$ independent of $N$ such that for $\varepsilon > 0$ sufficiently small:
\begin{equation}
\label{eq:decsmall1}
 \mathbb{P}\left(\left|\frac{\pi_{k, \infty}}{\pi_{k,0}} - 1 \right| > \varepsilon\right) \ge c.
\end{equation}
\item[(iii)]
For $n_{k,0} = o(1)$ (i.e. $\pi_{k,0} = o(1/N)$), we have $\var \left(\frac{\pi_{k, \infty}}{\pi_{k,0}}\right) \to \infty$ as $N \to \infty$.
\end{enumerate}
\item
If $R_t = \Theta(t^{-\alpha})$ for $\alpha > \frac{1}{2}$, then
\begin{enumerate}[itemsep = 3 pt]
\item[(i)]
For $n_{k,0} > 0$ such that $N n_{k,0} \to \infty$ as $N \to \infty$ (i.e. $N^2 \pi_{k,0} \to 0$), we have for each $\varepsilon > 0$ and each $t \ge 1$ or $t = \infty$:
\begin{equation}
\label{eq:decinter2}
\mathbb{P}\left(\left|\frac{\pi_{k, t}}{\pi_{k,0}} - 1\right| > \varepsilon \right) \le \frac{\sum_{t \ge 1} R_t^2}{\varepsilon^2 N n_{k,0}},
\end{equation}
which converges to $0$ as $N \to \infty$.
\item[(ii)]
For $n_{k,0} = \Theta(1/N)$ (i.e. $\pi_{k,0} = \Theta(1/N^2)$), we have $\var \left(\frac{\pi_{k, \infty}}{\pi_{k,0}}\right) = \Theta(1)$.
\item[(iii)]
For $n_{k,0} = o(1/N)$ (i.e. $\pi_{k,0} = o(1/N^2)$), we have $\var \left(\frac{\pi_{k, \infty}}{\pi_{k,0}}\right) \to \infty$ as $N \to \infty$.
\end{enumerate}
\item
If $R_t = \Theta(t^{-\alpha})$ for $\alpha < \frac{1}{2}$, then
\begin{enumerate}[itemsep = 3 pt]
\item[(i)]
For $n_{k,0} > 0$ such that $N^{\frac{\alpha}{1-\alpha}} n_{k,0} \to \infty$ as $N \to \infty$ (i.e. $N^{\frac{1}{1-\alpha}} \pi_{k,0} \to 0$), 
there is $C > 0$ independent of $t$ and $N$ such that for each $\varepsilon > 0$ and each $t \ge 1$ or $t = \infty$:
\begin{equation}
\label{eq:decinter3}
\mathbb{P}\left(\left|\frac{\pi_{k, t}}{\pi_{k,0}} - 1\right| > \varepsilon \right) \le \frac{C}{N^{\frac{\alpha}{1-\alpha}} n_{k,0}},
\end{equation}
which converges to $0$ as $N \to \infty$.
\item[(ii)]
For $n_{k,0} = \Theta(N^{-\frac{\alpha}{1-\alpha}})$ (i.e. $\pi_{k,0} = \Theta(N^{-\frac{1}{1-\alpha}})$), we have $\var \left(\frac{\pi_{k, \infty}}{\pi_{k,0}}\right) = \Theta(1)$.
Moreover, there is $c > 0$ independent of $N$ such that for $\varepsilon > 0$ sufficiently small:
\begin{equation}
\label{eq:decsmall3}
 \mathbb{P}\left(\left|\frac{\pi_{k, \infty}}{\pi_{k,0}} - 1 \right| > \varepsilon\right) \ge c.
\end{equation}
\item[(iii)]
For $n_{k,0} = o(N^{-\frac{\alpha}{1-\alpha}})$ (i.e. $\pi_{k,0} = o(N^{-\frac{1}{1-\alpha}})$), we have $\var \left(\frac{\pi_{k, \infty}}{\pi_{k,0}}\right) \to \infty$ as $N \to \infty$.
\end{enumerate}
\end{enumerate}
\end{theorem}

\quad The proof of Theorem \ref{thm:2} is given in Appendix \ref{scB}.
The theorem distinguishes three ways that the reward function decreases,
leading to different thresholds in the phase transition to categorize large, medium and small investors.
Part ($1$) assumes that the reward function decreases to a nonzero value.
In this case, the threshold to identify large, medium and small investors is $n_{k,0} = \Theta(1)$, 
which is the same as that of the PoS protocol with a constant reward.
This may not be surprising,
since the underlying dynamics is not much different from the one with a constant reward in the long run. 
For large investors, the ratio $\pi_{k, \infty}/\pi_{k,0}$ concentrates at $1$;
while for medium investors there is the {\em anti-concentration} bound \eqref{eq:decsmall1},
indicating that the evolution of a medium investor's share is no longer stable, and may be volatile. 
For small investors, we know that the variance of $\pi_{k, \infty}/\pi_{k,0}$ is large.

\quad Part ($2$) considers a fast decreasing reward $R_t = \Theta(t^{-\alpha})$ with $\alpha > 1/2$.
In this case, the threshold to identify different types of investors is $n_{k,0} = 1/N$, 
which is independent of the exact decreasing rate of $R_t$.
For large investors, the ratio $\pi_{k, \infty}/\pi_{k,0}$ concentrates at $1$;
while for medium (resp. small) investors, 
the variance of $\pi_{k, \infty}/\pi_{k,0}$ is bounded (resp. tends to infinity).

\quad Finally, part ($3$) deals with a slow decreasing reward $R_t = \Theta(t^{-\alpha})$ with $\alpha < 1/2$.
In contrast with the previous cases, 
the threshold $n_{k,0} = \Theta(N^{-\frac{\alpha}{1 - \alpha}})$ to identify different types of investors 
depends on the decreasing rate of $R_t$.
When $\alpha \to 0$, the threshold becomes $\Theta(1)$ which is consistent with that in part ($1$), 
and when $\alpha \to 1/2$, the threshold becomes $\Theta(1/N)$  which agrees with that in part ($2$).
For large investors, the ratio $\pi_{k, \infty}/\pi_{k,0}$ concentrates at $1$;
while for medium investors there is the anti-concentration bound \eqref{eq:decsmall3}.
For small investors, the variance of $\pi_{k, \infty}/\pi_{k,0}$ is large.

\quad Note that the statements for medium and small investors in all three cases
are weaker than those in Theorem \ref{thm:1}.
This is due to the fact the exact distributions of $\pi_{k,\infty}$ are not available for general rewarding schemes.
See Appendix \ref{scE1} for numerical illustrations concerning Theorem \ref{thm:2}.

\subsection{Increasing reward}

We consider the increasing reward of form $R_t = \rho N_{t-1}^{\gamma}$ for some $\rho > 0$ and $\gamma > 0$ (e.g. EOS).
The following theorem shows that the shares evolve very differently in the PoS protocol 
with a geometric reward vs a sub-geometric one. 

\begin{theorem}
\label{thm:3}
Assume that the reward $R_t = \rho N_{t-1}^{\gamma}$ for some $\rho > 0$ and $\gamma > 0$.
\begin{enumerate}[itemsep = 3 pt]
\item
If $\gamma > 1$, 
then $\pi_{k, \infty} \in \{0,1\}$ almost surely with
\begin{equation}
\label{eq:incextreme}
\mathbb{P}(\pi_{k, \infty} = 1) = \pi_{k,0}, \quad \mathbb{P}(\pi_{k, \infty} = 0) = 1 - \pi_{k,0}
\end{equation}
\item
If $\gamma < 1$, then
\begin{enumerate}[itemsep = 3 pt]
\item[(i)]
For $n_{k,0} = f(N)$ such that $f(N)/N^{\gamma} \to \infty$ as $N \to \infty$ (i.e. $\pi_{k,0})/N^{\gamma - 1} \to \infty$), we have for each $\varepsilon > 0$ and each $t \ge 1$ or $t = \infty$:
\begin{equation}
\label{eq:incinter}
\mathbb{P}\left(\left|\frac{\pi_{k, t}}{\pi_{k,0}} - 1\right| > \varepsilon \right) \le \frac{\rho N^{\gamma}}{(1 - \gamma) n_{k,0} \varepsilon^2},
\end{equation}
which converges to $0$ as $N \to \infty$.
\item[(ii)]
For $n_{k,0} = \Theta(N^{\gamma})$ (i.e. $\pi_{k,0} = \Theta(N^{\gamma-1})$), we have $\var \left(\frac{\pi_{k, \infty}}{\pi_{k,0}}\right) = \Theta(1)$.
Moreover, there exists $c > 0$ independent of $N$ such that for $\varepsilon > 0$ sufficiently small:
\begin{equation}
\label{eq:incsmall}
 \mathbb{P}\left(\left|\frac{\pi_{k, \infty}}{\pi_{k,0}} - 1 \right| > \varepsilon\right) \ge c.
\end{equation}
\item[(iii)]
For $n_{k,0} = o(N^{\gamma})$ (i.e. $\pi_{k,0} = o(N^{\gamma-1})$), we have $\var \left(\frac{\pi_{k, \infty}}{\pi_{k,0}}\right) \to \infty$ as $N \to \infty$.
\end{enumerate}
\end{enumerate}
\end{theorem}

\quad The proof of Theorem \ref{thm:3} is given in Appendix \ref{scC}.
The theorem deals with two increasing reward schemes: a geometric reward and a sub-geometric one.
Part ($1$) considers a geometric reward, 
and shows that with probability one, all the shares will eventually go to one investor in such a way that
\begin{equation*}
\mathbb{P}(\pi_k = 1 \mbox{ and } \pi_j = 0 \mbox{ for all } j \ne k) = \pi_{k,0}, \quad k \in [K].
\end{equation*}
We call this chaotic centralization because 
the underlying dynamics will lead to the dictatorship, but the dictator is selected in a completely random manner.
The investor with the largest initial coins will have the highest chance to control the PoS blockchain with a geometric reward.

\quad Part ($2$) considers a polynomial reward $R_t = \Theta(t^{\frac{1}{1-\gamma}})$ for $\gamma < 1$.
In this case, the threshold to distinguish large, medium and small investors is $n_{k,0} = \Theta(N^\gamma)$.
For large investors, the ratio $\pi_{k, \infty}/\pi_{k,0}$ concentrates at $1$;
while for medium investors there is the anti-concentration bound \eqref{eq:incsmall}.
For small investors, the variance of $\pi_{k, \infty}/\pi_{k,0}$ explodes.
See Appendix \ref{scE2} for numerical illustrations related to Theorem \ref{thm:3}.

\section{Infinite population model}
\label{sc3}

\quad This section is concerned with modeling the PoS protocol which allows for an ``infinite'' number of investors in response to  the growing popularity of blockchains.
In Section \ref{sc31}, we start with a discrete infinite population model as a direct extension to the P\'olya urn studied in Section \ref{sc2}.
In Section \ref{sc32}, we consider an infinite population model sampled from a continuous space.
Combining the ideas from these two subsections, 
we propose and analyze a dynamical population model for the PoS protocol in Section \ref{sc33}.

\subsection{Discrete infinite population}
\label{sc31}

In Section \ref{sc2}, we consider the P\'olya urn model with a finite number of investors, i.e. $K < \infty$, 
and study the limiting behavior of the ratio $\pi_{k,t}$ or $\pi_{k,\infty}$ as the initial coin offerings $N \to \infty$.
Here we work directly with $K = \infty$ investors to model the PoS protocol.

\quad We give a formal description of the model. 
The P\'olya urn model with an infinite population is just as described in Section \ref{sc2} but for $K = \infty$.
More precisely, there are a countably infinite number of investors indexed by $\mathbb{N}$.
Let $n_{k, t}$, $k \in \mathbb{N}$ be investor $k$'s endowment of coins at time $t$
with $\sum_{k = 1}^{\infty} n_{k,t} = N_t$.
At time $t+1$, investor $k$ is selected at random among all investors with probability $\pi_{k,t} := n_{k,t}/N_t$.
Once selected, the investor receives a deterministic reward of $R_t \ge 0$.
This way, the equations \eqref{eq:share0}--\eqref{eq:Dshare} hold for $K = \infty$,
and there are infinitely many small investors whose initial coins are $n_{k,0} = o(1)$ (i.e. $\pi_{k,0} = o(1/N)$).
The investor shares are given by a vector of infinite length 
$(\pi_{k,t}, \, k \in \mathbb{N}_+)$
Again by the martingale convergence theorem,
\begin{equation}
(\pi_{1,t}, \pi_{2,t}, \ldots) \to (\pi_{1,\infty}, \pi_{2,\infty}, \ldots) \quad \mbox{as } t \to \infty,
\end{equation}
where $(\pi_{k,\infty}, \, k \in \mathbb{N}_+)$ is random probability measure on $\mathbb{N}_+$. 

\quad We have seen in Theorem \ref{thm:1} that for a finite number of $K$ investors, if the reward $R_t$ is constant, the limiting investor shares $(\pi_{1,\infty}, \ldots, \pi_{K,\infty})$ have the Dirichlet distribution.
For the P\'olya urn model with an infinite population, 
one important problem is to understand limiting shares $(\pi_{1,\infty}, \pi_{2, \infty}, \ldots)$ 
under constant or more general rewarding schemes.
To this end, we recall the definition of the {\em Dirichlet-Ferguson measure}, or simply {\em Dirichlet measure} 
which was introduced by \cite{Fer73, BM73} in the context of nonparametric Bayesian analysis.

\begin{definition}
\label{def:Dirichlet}
Let $S$ be a Polish space with Borel $\sigma$-field $\mathcal{S}$, 
and let $\mu$ be a positive measure on $(S, \mathcal{S})$ with $0 < \mu(S) < \infty$.
We say that $F$ has $\Dir(\mu)$ distribution if $F$ is a random distribution on $S$
such that for every measurable partition $B_1, \ldots, B_k$ of $S$, the random vector $(F(B_1), \ldots, F(B_k))$ has $\Dir(\mu(B_1), \ldots, \mu(B_k))$ distribution.
\end{definition}

\quad The following theorem studies the evolution of shares in a PoS protocol modeled by the P\'olya urn with an infinite population.
\begin{theorem}
\label{thm:4} 
Assume that the coin reward $R_t \equiv R > 0$. Then the investor shares have a limiting distribution
\begin{equation}
\label{eq:limDirinf}
(\pi_{1, \infty}, \pi_{2, \infty}, \ldots) \stackrel{d}{=} \Dir(\mu),
\end{equation}
where $\mu$ is a positive measure on $\mathbb{N}$ with $\mu(\{k\}) = \frac{n_{k,0}}{R}$, $k \in \mathbb{N}$.
Moreover, for each investor or a finite number of investors in the same category,
the results in Theorems \ref{thm:1}, \ref{thm:2} and \ref{thm:3} hold under the corresponding rewarding schemes.
\end{theorem}

\quad The proof of Theorem \ref{thm:4} is given in Appendix \ref{scD1}. 
Note that if there are a finite number of $K$ investors, 
the measure $\mu$ is then supported on $[K]$.
This recovers the identity in distribution \eqref{eq:limDirichlet} in Theorem \ref{thm:1}.

\subsection{Infinite population from continuum}
\label{sc32}

We consider an urn model with an infinite population which is sampled from a continuous space -- 
we call it a {\em PoS feature model}.
The motivation comes from understanding the influence induced by common features among investors in the PoS protocol.
The influence of a particular feature is measured by the total shares that investors having this feature own.
In many generic cases, features are represented or approximated by elements in a continuous sample space, 
e.g. geolocation of an investor, market experience of an investor measured in time, index assessing the level of risk aversion of an investor, and so on.
Without loss of generality, we abstract the feature space as the unit interval $S = [0,1]$.

\quad The PoS feature model is inspired from the Blackwell-MacQueen construction of a P\'olya urn on general state spaces as described in Lemma \ref{lem:BMscheme}.
At each time $t \ge 1$, an investor with some feature $X_t \in S = [0,1]$ is selected to receive a deterministic reward $R_t \ge 0$.
Now we specify the selection rule over the time. 
At time $t = 0$, the initial coin offering is $N_0 = N$, and these coins are distributed among the investors with features in $S = [0,1]$ according to a diffuse probability measure $\nu$ on $S = [0,1]$.
That is, the number of coins owned by investors with features in $[x, x+dx]$ is $N \nu(x) dx$.
At time $t = 1$, an investor with feature $X_1$ is selected by the rule
\begin{equation}
\label{eq:selection1}
\mathbb{P}(X_1 \in \cdot) = \nu(\cdot),
\end{equation}
and then receives a reward $R_1$.
So the total number of coins becomes $N_1 = N + R_1$. 
At time $t = 2$, an investor with feature $X_2$ is selected with probability $\mathbb{P}(X_2 \in \cdot) = (R_1 \delta_{X_1}(\cdot) + N \nu(\cdot))/N_1$. 
More generally, at time $t \ge 2$ an investor with feature $X_t$ is selected by the rule:
\begin{equation}
\label{eq:selectiont}
\mathbb{P}(X_t \in \cdot|X_1, \ldots, X_{t-1}) = \frac{\sum_{n = 1}^{t-1}R_n \delta_{X_n}(\cdot)}{N + \sum_{n = 1}^{t-1} R_n} + \frac{N \nu(\cdot)}{N + \sum_{n = 1}^{t-1} R_n}.
\end{equation}

\quad The main difference between the PoS feature model \eqref{eq:selection1}--\eqref{eq:selectiont} and the urn models discussed in the previous sections is that
there are uncountably many features of the investors for selection,
but there are only a countable number of investors selected at time $t = 1,2, \ldots$
Two questions arise naturally:
($1$). How to label the features of the investors selected by the feature model \eqref{eq:selection1}--\eqref{eq:selectiont}?
($2$). What are the limiting shares corresponding to these features?

\quad These problems are closely related to the problem of species sampling and exchangeable partitions studied in \cite{Pitman95, Pitman96}.
Let us spell out in the PoS setting as follows.
For (1) the simplest way to label the features among the selected investors is by their order of appearance. 
For $j \ge 1$, denote $\widetilde{X}_j$ as the $j^{th}$ feature to appear in the sequence of $X_1, X_2, \ldots$
Let $M_1:= 1$ and 
$M_j: = \inf\{n: n > M_{j-1}, X_n \notin \{X_1, \ldots, X_{n-1}\}\}$ 
for $j \ge 2$,
with the convention $\inf \emptyset = \infty$. 
So $M_j$ is the index at which the $j^{th}$ feature appears for the first time, and $\widetilde{X}_j = X_{M_j}$ on the event $\{M_j < \infty\}$.
For instance, if 
$(X_1, X_2, \ldots) = (0.1, 0.1, 0.3, 0.2, 0.2, 0.3, 0.1, 0.4, \ldots)$,
then $M_1 =1$, $M_2 = 3$, $M_3 = 4$, $M_4 = 8, \ldots$ and $\widetilde{X}_1 = 0.1$, $\widetilde{X}_2 = 0.3$,
$\widetilde{X}_3 = 0.2$, $\widetilde{X}_4 = 0.4, \ldots$
For general rewards $R_t$, it seems challenging to specify the limiting distribution of shares $(\pi_{\widetilde{X}_1, \infty}, \pi_{\widetilde{X}_2, \infty}, \ldots)$.
One exception is for the constant reward $R_t \equiv  R$ as stated in the following theorem.

\begin{theorem}
\label{thm:5}
Assume that the coin reward $R_t \equiv  R > 0$. 
Let $\widetilde{X}_1, \widetilde{X}_2, \ldots$ be the features appearing in the order of appearance 
of the PoS feature model specified by \eqref{eq:selection1}--\eqref{eq:selectiont}. 
Then there is the stick-breaking representation for the limiting shares:
\begin{equation}
\label{eq:stickbreaking}
\pi_{\widetilde{X}_j} = \left[\prod_{i = 1}^{j-1} (1 - W_i)\right] W_j \quad \mbox{for } j \ge 1,
\end{equation}
where $W_1, W_2, \ldots$ are independent and identically distributed as $\bet(1, N/R)$.
Moreover, let $K_t: = \sup\{j: M_j \le t\}$ be the number of features appeared among the first $t$ selected investors. 
Then $K_t/\log t \to N/R$ almost surely.
\end{theorem}

\quad The proof of Theorem \ref{thm:5} is given in Appendix \ref{scD2}.
Theorem \ref{thm:5} shows that whatever the initial distribution of features is, 
the PoS feature model \eqref{eq:selection1}--\eqref{eq:selectiont} with a constant reward yields a limiting share distribution on a countable number of features.
This distribution, which only depends on the initial coin offerings $N$ and the reward $R$, is known as the {\em Griths-Engen-McCloskey} (GEM) distribution (\cite{Ewens90}).

\quad We may also consider more general PoS feature models.
One such instance is when the selection rule relies on the history of the features of previously selected investors.
For $j \ge 1$, let
$N_{jt}:=\sum_{n = 1}^t 1(X_n = \widetilde{X}_j, M_j < \infty)$
be the number of times that the investors with feature $j$ (in the order of appearance) are selected up to time $t$,
and $\pmb{N}_t:= (N_{1t}, N_{2t}, \ldots)$ be the vector of counts of various features of the investors up to time $t$.
We can regroup the investors according to their features, and rewrite the selection rule \eqref{eq:selectiont} as
\begin{equation}
\label{eq:selectiont2}
\mathbb{P}(X_t \in \cdot|X_1, \ldots, X_{t-1}, K_{t-1} = k) = \sum_{j = 1}^k \frac{N_{jt-1}}{N/R + t- 1} 1(\widetilde{X}_j \in \cdot) + \frac{N/R}{N/R +t-1} \nu(\cdot).
\end{equation}
Here we look for general selection rules of form
$\mathbb{P}(X_1 \in \cdot) = \nu(\cdot)$ and for $t \ge 2$,
\begin{equation}
\label{eq:selectiontgen}
\mathbb{P}(X_t \in \cdot|X_1, \ldots, X_{t-1}) = \sum_{j = 1}^k p_j(\pmb{N}_{t-1}) 1(\widetilde{X}_j \in \cdot) + p_{k+1}(\pmb{N}_{t-1}) \nu(\cdot),
\end{equation}
for some functions $p_j$, $j = 1,2, \ldots$ defined on $\cup_{k =1}^{\infty} \mathbb{N}_+^k$.
The meaning of the selection rule \eqref{eq:selectiontgen} is as follows: given the histogram $\pmb{N}_{t-1}$ of $k$ features of investors selected from time $1$ to time $t-1$, an investor with feature $j$ is selected with probability $p_j(\pmb{N}_{t-1})$ for $1 \le j \le k$, and an investor with a new feature $k+1$ is selected with probability $p_{k+1}(\pmb{N}_{t-1})$.
It is easily seen the selection rule \eqref{eq:selectiont2} is a special case of the general rule \eqref{eq:selectiontgen} with
$p_j(n_1, \ldots, n_k) = \frac{n_j}{N/R + t-1} 1(1 \le j \le k) + \frac{N/R}{N/R +t-1} 1(j = k-1)$,
with $\sum_{j=1}^k n_j = t-1$.
A closely related selection rule is defined by the functions
$p_j(n_1, \ldots, n_k) = \frac{n_j - \alpha}{N/R + t-1} 1(1 \le j \le k) + \frac{N/R + k \alpha}{N/R +t-1} 1(j = k-1)$,
for some $\alpha \ge 0$.
In this case, the limiting share distribution also has the stick-breaking representation \eqref{eq:stickbreaking} 
with $W_1, W_2, \ldots$ independent, and $W_k$ distributed as $\bet(1-\alpha, N/R + k \alpha)$.
This is known as the {\em Pitman-Yor distribution} (\cite{Pitman96, PY97}). 

\quad In general, we need the following condition on $p_j$ to define the selection rule \eqref{eq:selectiontgen}:
$p_j(\pmb{n}) \ge 0$ and $ \sum_{j = 1}^{|\pmb{n}|_0 + 1} p_j(\pmb{n}) = 1$ for $ \pmb{n} \in \cup_{k =1}^{\infty} \mathbb{N}_+^k$,
where $|\pmb{n}|_0$ is the number of nonzero entries in $\pmb{n}$.
\cite{Pitman96} provides an additional condition on $p_j$ in terms of the {\em exchangeable partition probability functions} (EPPFs) so that 
the sequence $X_1, X_2, \ldots$ specified by the selection rule \eqref{eq:selection1}--\eqref{eq:selectiontgen} is exchangeable, and thus the limiting share distribution is well-defined. 
Such a sequence $X_1, X_2, \ldots$ is called the {\em species sampling sequence}, see also \cite{HP00} for related discussions.

\subsection{From finite to infinite population}
\label{sc33}

In this final subsection, we propose a dynamical approach to model the PoS protocol.
We start with a finite number of investors, and then at each time a new investor may come into the market,
which evolves to an infinite number of investors.
This combines the ideas from previous sections, especially the Blackwell-MacQueen urn model.

\quad We proceed to the description of the model.
At time $t = 0$ there are $K$ investors indexed by $[K]$.
These investors are initial capitalists in the blockchain network, so they play a very important role in the PoS protocol.
For $k \in [K]$, let $n_{k,0}$ be the initial endowment of coins of investor $k$, and 
$\pi_{k,0}: = n_{k,0}/N$ with $N = \sum_{k=1}^K n_{k,0}$ be the corresponding shares. 
At time $t = 1$, there are two possibilities: either one of these $K$ initial investors are selected, or a new investor is selected from the population. 
This is realistic since many cryptocurrencies are initially owned by a handful of coin miners or venture capitalists,
and then their shares will be diluted by new investors over the time. 
Since the population is large, we approximate the population space by the unit interval $S = [0,1]$, and a new investor is selected from $S = [0,1]$ by a diffuse probability measure $\nu$ as in Section \ref{sc32}.
We also introduce a {\em dilution parameter} $\theta > 0$, which is the weight that a new investor is introduced to the blockchain network.
More precisely,
\begin{itemize}[itemsep = 3 pt]
\item
For each $k \in [K]$, investor $k$ is selected with probability $\frac{n_{k,0}}{N + \theta}$.
\item
A new investor with ``index' in $(0,1)$ is selected with probability $\frac{\theta \nu(\cdot)}{N+\theta}$.
\end{itemize}
By letting $X_1$ be the index of the investor selected at time $1$, we have
\begin{equation}
\label{eq:selection1f}
\mathbb{P}(X_1 \in \cdot) = \sum_{k = 1}^K \frac{n_{k,0}}{N+\theta} \delta_k(\cdot) + \frac{\theta \nu(\cdot)}{N + \theta},
\end{equation}
and the selected investor receives a deterministic reward $R_1 > 0$.
More generally, at time $t$ the selection rule is given by 
\begin{equation}
\label{eq:selectiontf}
\mathbb{P}(X_t \in \cdot |X_1, \ldots, X_{t-1}) =  \sum_{k = 1}^K \frac{n_{k,t-1}}{N_{t-1} + \theta} \delta_k(\cdot) + \sum_{X_n \in (0,1)} \frac{n_{X_n, t-1}}{N_{t-1} + \theta} \delta_{X_n}(\cdot) + \frac{\theta \nu(\cdot)}{N_{t-1} + \theta},
\end{equation}
where $n_{k,t}$ is the number of coins that investor $k$ owns at time $t$, 
$n_{X_n, t}$ is the number of coins that a new investor with index $X_n \in (0,1)$ owns at time $t$, and 
$N_{t} = N + \sum_{n = 1}^t R_n$ is the total number of coins up to time $t$.
There are three terms on the right side of the selection rule \eqref{eq:selectiontf}:
the first one comes from the $K$ initial investors, 
the second term is from the new investors previously entering the market, 
and the third term is the probability that a new investor is introduced.
The following theorem studies how the shares of the initial capitalists are diluted over the time.

\begin{theorem}
\label{thm:6}
For $k \in [K]$, let $\pi_{k,t}$ be the shares of investor $k$ at time $t$ under the PoS protocol \eqref{eq:selection1f}--\eqref{eq:selectiontf}.
Then $(\pi_{k,t}, \, t \ge 0)$ is a supermartingale. 
Consequently, the shares of initial investors $(\pi_{1,t}, \ldots, \pi_{K,t}) \to (\pi_{1,\infty}, \ldots, \pi_{K,\infty})$ almost surely for some
random sub-probability distribution $(\pi_{1,\infty}, \ldots, \pi_{K,\infty})$.
\begin{enumerate}[itemsep = 3 pt]
\item
Assume that the coin reward $R_t$ is decreasing with $R_t \ge R_{t+1}$ for each $t \ge 0$.
\begin{enumerate}[itemsep = 3 pt]
\item[(i)]
If $\lim_{t \to \infty} R_t = R > 0$ or $R_t = \Theta(t^{-\alpha})$ for $\alpha < 1$, we have
$0 < \mathbb{E}\left( \frac{\pi_{k, \infty}}{\pi_{k,0}}\right) < 1$,
and $\lim_{N \to \infty} \mathbb{E}\left( \frac{\pi_{k, \infty}}{\pi_{k,0}}\right) = 1$.
\item[(ii)] 
If $R_t = \Theta(t^{-\alpha})$ for $\alpha > 1$, we have $\pi_{k, \infty} = 0$ almost surely.
\end{enumerate}
\item
Assume that the coin reward $R_t = \rho N_{t-1}^{\gamma}$ for $\rho, \gamma > 0$.
We have $0 < \mathbb{E}\left( \frac{\pi_{k, \infty}}{\pi_{k,0}}\right) < 1$,
and $\lim_{N \to \infty} \mathbb{E}\left( \frac{\pi_{k, \infty}}{\pi_{k,0}}\right) = 1$.
\item
Assume that the coin reward $R_t = R > 0$. 
We have $\pi_{k, \infty} \stackrel{d}{=} \bet\left(\frac{n_{k,0}}{R}, \frac{N + \theta - n_{k,0}}{R}\right)$.
Consequently, the results in Theorem \ref{thm:1} hold.
\end{enumerate}
\end{theorem}

\quad The proof of Theorem \ref{thm:6} is given in Appendix \ref{scD3}. 
The theorem shows that for the dynamical PoS model \eqref{eq:selection1f}--\eqref{eq:selectiontf}, the shares of initial investors will decrease over the time. 
If the coin reward does not decay too fast, i.e. $R_t \gg t^{-1}$, the expectation of the ratio $\pi_{k, \infty}/\pi_{k,0}$ tends to $1$ as the initial coin offering $N$ is large.
On the contrary, if the coin reward decays very fast, i.e. $R_t \ll t^{-1}$, the initial investors' shares will be eventually diluted to zero.
More is known if the coin reward is constant. 
As is clear in the proof, the selection rule \eqref{eq:selectiontf} as a probability measure converges 
with probability one to a random discrete probability distribution $F \stackrel{d}{=} \Dir ((\sum_{k = 1}^K n_{k,0} \delta_k + \theta \nu)/R)$, and given $F$ the indices of selected investors, i.e. $X_1, X_2, \ldots$ are independent and identically distributed as $F$.
Recall that $F$ has the representation $\sum_{t =1}^{\infty} P_t \delta_{Z_t}$ where $(P_1, P_2, \ldots)$ has $\GEM(\frac{N + \theta}{R})$ distribution, and $Z_1, Z_2, \ldots$ are independent and identically distributed as $(\sum_{k = 1}^K n_{k,0} \delta_k + \theta \nu)/(N + \theta)$ and also independent of $(P_1, P_2, \ldots)$.
Thus, the indices of new investors (in order of their appearance) are just independent and identically distributed as $\nu$, and the limiting expectation of their total shares is $\frac{\theta}{N + \theta}$.
So the larger the initial coin offerings $N$ are, the less influence of new investors have.

\section{Conclusion}
\label{sc4}

\quad In this paper, we study the evolution of investor shares in a PoS blockchain under various rewarding schemes and for different types of investors.
In contrast with the previous works where only large investors are considered, 
we take the heterogeneity of investors into account, 
and show that medium to small investors may suffer from share instability 
-- their shares may be volatile, or even shrink to zero in the long run.
This leads to the phase transition phenomenon, 
where thresholds for stability vs instability are characterized under different rewarding schemes.
In particular, for the PoS protocol with a geometric reward we observe chaotic centralization; 
that is, all the shares will go to one investor in a random manner.
In response to the increasing activities in blockchain networks, 
we also propose and analyze a dynamical population model for the PoS protocol which allows the number of investors to grow over the time. 
Our quantitative analysis also provides guidance on the choice of rewarding schemes so that decentralization is indeed implemented in a PoS blockchain.

\quad There are a few directions to extend this work.
For instance, one can study the trading incentive in the dynamical population setting.
This requires incorporating a game-theoretic component to the analysis,
and formulating a suitable reward optimization problem.
Another problem is to consider other types of urn dynamics to model the voting rule in the PoS protocol, 
e.g. square voting rule (\cite{Pen46}),
and study the problems of long time stability and reward incentives.

\bigskip
{\bf Acknowledgement:} 
We thank Alex Y. Wu for the help with the numerical experiments and the careful reading of the manuscript. 
We thank David Yao for helpful discussions, and Agostino Capponi and Jing Huang for various pointers to the literature.
We gratefully acknowledges financial support through an NSF grant DMS-2113779 and through a start- up grant at Columbia University.

\appendix
\section{Proof of Theorem \ref{thm:1}}
\label{scA}
\begin{proof}
The fact that the investor shares $(\pi_{1,t}, \ldots, \pi_{K,t})$ converges almost surely as $t \to \infty$ to $(\pi_{1,\infty}, \ldots, \pi_{K,\infty})$ with the Dirichlet distribution \eqref{eq:limDirichlet} is a classical result of P\'olya urn, see e.g.\cite[Section 6.3]{JK77}.

($i$) Since $(\pi_{k,t}, \, t \ge 0)$ is a martingale, we have $\mathbb{E}(\pi_{k,t}) = \pi_{k,0}$.
It follows from  \cite[Corollary 3.1]{M09} or Lemma 2.2 of \cite[Lemma 2.2]{GR13} that 
\begin{equation}
\label{eq:vart}
\var(\pi_{k,t}) = \frac{R^2}{(Rt+N)^2}\left(\frac{R}{N+R}t^2 + \frac{N}{N+R} t \right) \pi_{k,0}(1-\pi_{k,0}).
\end{equation}
Specializing \eqref{eq:vart} to $t = \infty$, we obtain $\var(\pi_{k,\infty}) = \frac{R}{N+R} \pi_{k,0}(1-\pi_{k,0})$ which recovers Corollary 1 of \cite{RS21}.
By applying Chebyshev's inequality, we get
\begin{align}
\mathbb{P}\left(\left|\frac{\pi_{k, t}}{\pi_{k,0}} - 1\right| > \varepsilon \right) & \le \frac{\var(\pi_{k,t})}{\varepsilon^2 \pi_{k,0}^2} \notag \\
& = \frac{1 - \pi_{k,0}}{\varepsilon^2 \pi_{k,0}} \left(\frac{R^3 t^2}{(Rt + N)^2(N+R)} + \frac{R^2 N t}{(Rt+N)^2 (N+R)} \right)  \notag \\
& \le  \frac{1 - \pi_{k,0}}{\varepsilon^2 \pi_{k,0}} \left(\frac{R}{N+R} + \frac{R}{4(N+R)} \right), \label{eq:varineq}
\end{align}
which leads to \eqref{eq:coninter} by noting that $\pi_{k,0} = \frac{f(N)}{N}$ and $1 - \pi_{k,0} \le 1$.

($ii$) Since $(\pi_{1,\infty}, \ldots, \pi_{K,\infty}) \stackrel{d}{=} \Dir\left(\frac{n_{1,0}}{R}, \ldots,\frac{n_{K,0}}{R}\right)$, we have 
$\pi_{k, \infty} \stackrel{d}{=} \bet(\frac{n_{k,0}}{R}. \frac{N - n_{k,0}}{R})$.
Now $n_{k,0} = \Theta(1)$, by standard limit theorem we get $N \pi_{k,\infty}/R \stackrel{d}{\longrightarrow} \gamma\left( \frac{n_{k,0}}{R}\right)$.
This implies the convergence in distribution \eqref{eq:consmall}.

($iii$) By ($ii$) we have $\frac{\pi_{k, \infty}}{\pi_{k,0}} \stackrel{d}{=} \frac{R}{n_{k,0}} \gamma\left(\frac{n_{k,0}}{R}\right)$.
Let $a = \frac{n_{k,0}}{R}$ so $a \to 0$ as $N \to \infty$ by hypothesis.
Then
\begin{equation*}
\mathbb{P}\left(\frac{\pi_{k,\infty}}{\pi_{k,0}} < \varepsilon \right) = \mathbb{P} \left( \frac{\gamma(a)}{a} \le \varepsilon \right) 
= \frac{1}{\Gamma(a)} \int_0^{a\varepsilon} x^{a-1}e^{-x}dx.
\end{equation*}
Note that $\int_0^{a\varepsilon} x^{a-1}e^{-x}dx \ge e^{-a \varepsilon} \frac{(a \varepsilon)^a}{a}$, and $a \Gamma(a) \to 1$ as $a \to 0$.
As a result, there exists a constant $C > 0$ (independent of $a$) such that
\begin{equation*}
\mathbb{P}\left(\frac{\pi_{k,\infty}}{\pi_{k,0}} < \varepsilon \right) \ge Ce^{-a \varepsilon} (a \varepsilon)^a = C \exp\left( -a \varepsilon + a \log a + a \log \varepsilon \right),
\end{equation*}
which converges to $1$ as $a \to 0$. This yields the convergence \eqref{eq:continy}.
\end{proof}

\section{Proof of Theorem \ref{thm:2}}
\label{scB}

\quad To prove Theorem \ref{thm:2}, we need a series of lemmas.
\begin{lemma}
\label{lem:generalmoment}
The variance of investor $k$'s share is
\begin{equation}
\label{eq:generalvart}
\var(\pi_{k,t}) = a_t \pi_{k,0}(1 - \pi_{k,0}),
\end{equation}
where 
\begin{equation}
\label{eq:at}
a_1 = \left(\frac{R_1}{N_1}\right)^2, \quad a_{t+1} = a_t + \left(\frac{R_{t+1}}{N_{t+1}}\right)^2(1-a_t).
\end{equation}
The third central moment of investor $k$'s share satisfy the relation
\begin{align}
\label{eq:generalthird}
\mu_3(\pi_{k,t+1}) = \mu_3(\pi_{k,t}) & + \frac{3 R_{t+1}^2}{N_{t+1}^2}\mathbb{E}[(\pi_{k,t} - \pi_{k,0}) \pi_{k,t}(1 - \pi_{k,t})] \notag \\
& + \frac{R_{t+1}^3}{N_{t+1}^3}\mathbb{E}[\pi_{k,t}(1 - \pi_{k,t})(1 - 2 \pi_{k,t})],
\end{align}
and the fourth central moment of investor $k$'s share satisfy the relation
\begin{align}
\label{eq:generalfourth}
\mu_4(\pi_{k,t+1}) = \mu_4(\pi_{k,t}) & + \frac{6 R_{t+1}^2}{N_{t+1}^2}\mathbb{E}[(\pi_{k,t} - \pi_{k,0})^2 \pi_{k,t}(1 - \pi_{k,t})] \notag \\
& + \frac{4 R_{t+1}^3}{N_{t+1}^3}\mathbb{E}[(\pi_{k,t} - \pi_{k,0})\pi_{k,t}(1 - \pi_{k,t})(1 - 2 \pi_{k,t})]  \notag \\
& + \frac{R_{t+1}^4}{N_{t+1}^4}\mathbb{E}[\pi_{k,t}(1 - \pi_{k,t})(1 - 3 \pi_{k,t}+ 3 \pi_t^2)].
\end{align}
\end{lemma}
\begin{proof}
The formula \eqref{eq:generalvart}--\eqref{eq:at} for the variance of investor's share is read from \cite[Lemma A.2]{RS21}.
It is easily seen from \eqref{eq:Dshare} that
\begin{equation}
\label{eq:centmoment}
\pi_{k,t+1} - \pi_{k,0} = (\pi_{k,t} - \pi_{k,0}) + \frac{R_{t+1}}{N_{t+1}}(1_{S_{k, t+1}} - \pi_{k,t}). 
\end{equation}
Taking the third and the fourth moment on both sides of \eqref{eq:centmoment}, and using the binomial expansion yields the formulas \eqref{eq:generalthird}--\eqref{eq:generalfourth}.
\end{proof}

\begin{lemma}
\label{lem:asympnozero}
Assume that the reward $R_t$ is decreasing, i.e. $R_t \ge R_{t+1}$ for each $t \ge 0$, and that $\lim_{t \to \infty} R_t = \underline{R} > 0$.
\begin{enumerate}[itemsep = 3 pt]
\item
Let $a_t$ be defined by \eqref{eq:at}. We have
\begin{equation}
\label{eq:boundat}
\frac{(N - R_1)\underline{R}^2 t}{N(N + R_1)(N + R_1(1+t))} \le a_t \le \frac{R_1}{N}, \quad \mbox{for each } t \ge 1.
\end{equation}
\item
Let $\mu_3(\pi_{k,t})$ and $\mu_4(\pi_{k,t})$ be the third and the fourth central moment of investor $k$'s share satisfying \eqref{eq:generalthird}, \eqref{eq:generalfourth} respectively.
If $\pi_{k,0} = \mathcal{O}(1/N)$, there exist $C_3, C_4 > 0$ independent of $t$ and $N$ such that
\begin{equation}
\label{eq:bound34}
\mu_3(\pi_{k,t}) \le \frac{C_3 \pi_{k,0}}{N^2},  \quad \mu_4(\pi_{k,t}) \le \frac{C_4 \pi_{k,0}}{N^3} \quad \mbox{for each } t \ge 1.
\end{equation}
\end{enumerate}
\end{lemma}

\begin{proof}
(1) The upper bound $a_t \le \frac{R_1}{N}$ follows from \cite[Lemma A.4]{RS21}.
By \eqref{eq:at}, we get $a_{t + 1} \ge a_t + \left(1 - \frac{R_1}{N}\right) \left(\frac{R_{t+1}}{N_{t+1}}\right)^2$.
This implies that for $t \ge 1$,
\begin{align}
a_t & \ge \left(1 - \frac{R_1}{N}\right)\sum_{n = 1}^t \left(\frac{R_n}{N + \sum_{k=1}^n R_k} \right)^2 \notag\\
& \ge \left(1 - \frac{R_1}{N}\right)\sum_{n = 1}^t \left(\frac{R_n}{N + nR_1}\right)^2, \label{eq:atlow}
\end{align}
where the second inequality is due to the fact that $R_t$ is decreasing.
By abuse of language, let $(R_s, \,s \in [1, \infty))$ be the linear interpolation of $(R_t, \, t = 1,2, \ldots)$.
Since $s \to \frac{R_s}{N + sR_1}$ is decreasing, we deduce from \eqref{eq:atlow} that
\begin{equation*}
a_t  \ge \left(1 - \frac{R_1}{N}\right) \int_1^{t+1} \frac{R_s^2}{(N + sR_1)^2} ds  \ge \left(1 - \frac{R_1}{N}\right) \int_1^{t+1} \frac{\underline{R}^2}{(N + sR_1)^2} ds,
\end{equation*}
which yields the lower bound in \eqref{eq:boundat}.

(2) Note that $(\pi_{k,t} - \pi_{k,0}) \pi_{k,t}(1 - \pi_{k,t}) \le \pi_{k,t}^2(1 - \pi_{k,t}) \le \pi_{k,t}^2$.
Then for each $t \ge 1$,
\begin{align}
\label{eq:125}
\mathbb{E}[(\pi_{k,t} - \pi_{k,0}) \pi_{k,t}(1 - \pi_{k,t})]  \le \mathbb{E}(\pi_{k,t}^2) &= \pi_{k,0}^2 + \var(\pi_{k,t}) \notag \\
& \le \pi_{k,0}^2 + \frac{R_1 \pi_{k,0}}{N} \notag \\
& \le \frac{C \pi_{k,0}}{N} \quad \mbox{for some } C > 0,
\end{align}
where the second inequality follows from \eqref{eq:generalvart} and the upper bound in \eqref{eq:boundat}, and the last inequality is due to the fact that $\pi_{k,0} = \mathcal{O}(1/N)$.
Also for each $t \ge 1$,
\begin{equation}
\label{eq:126}
\mathbb{E}[\pi_{k,t}(1 - \pi_{k,t})(1 - 2 \pi_{k,t})] \le \mathbb{E}(\pi_{k,t}) = \pi_{k,0}.
\end{equation}
Combining \eqref{eq:generalthird}, \eqref{eq:125} and \eqref{eq:126}, we get
\begin{equation}
\label{eq:u3upper}
\mu_3(\pi_{k,t}) \le \frac{3C \pi_{k,0}}{N} \sum_{n = 1}^t \left(\frac{R_n}{N + \sum_{k=1}^n R_k} \right)^2 + \pi_{k,0} \sum_{n = 1}^t \left(\frac{R_n}{N + \sum_{k=1}^n R_k} \right)^3.
\end{equation}
Using the sum-integral trick as in \cite[Lemma A.4]{RS21} we deduce that 
$\sum_{n = 1}^t \left(\frac{R_n}{N + \sum_{k=1}^n R_k} \right)^k \le (R_1/N)^{k-1}$ for $k = 2, 3, \ldots$
So the bound \eqref{eq:u3upper} leads to
\begin{equation*}
\mu_3(\pi_{k,t}) \le 3CR_1 \frac{\pi_{k,0}}{N^2} + R_1^2 \frac{\pi_{k,0}}{N^2} \le \frac{C_3 \pi_{k,0}}{N^2} \quad \mbox{for some } C_3 > 0.
\end{equation*}
Similarly, we get from \eqref{eq:generalfourth} that
\begin{align*}
\mu_4(\pi_{k,t}) & \le \frac{C \pi_{k,0}}{N^2} \sum_{n = 1}^t \left(\frac{R_n}{N + \sum_{k=1}^n R_k} \right)^2 + \frac{C' \pi_{k,0}}{N}  \sum_{n = 1}^t \left(\frac{R_n}{N + \sum_{k=1}^n R_k} \right)^3  \\
& \qquad \qquad \qquad \qquad \qquad \qquad \qquad+ C''\pi_{k,0}  \sum_{n = 1}^t \left(\frac{R_n}{N + \sum_{k=1}^n R_k} \right)^4 \\
& \le C R_1 \frac{\pi_{k,0}}{N^3} + C'R_1^2 \frac{\pi_{k,0}}{N^3} + C'' R_1^3 \frac{\pi_{k,0}}{N^3} 
\le \frac{C_4 \pi_{k,0}}{N^3} \quad \mbox{for some } C_4 > 0.
\end{align*}
\end{proof}

\begin{lemma}
\label{lem:largerhalf}
Assume that the reward $R_t$ is decreasing, i.e. $R_t \ge R_{t+1}$ for each $t \ge 0$, and that $R_t = \Theta(t^{-\alpha})$ for $\alpha > \frac{1}{2}$.
\begin{enumerate}[itemsep = 3 pt]
\item
Let $a_t$ be defined by \eqref{eq:at}. We have
\begin{equation}
\label{eq:boundat2}
\frac{R_1^2}{(N+R_1)^2} \le a_t \le \frac{\sum_{t \ge 1} R_t^2}{N^2}, \quad \mbox{for each } t \ge 1.
\end{equation}
\item
Let $\mu_3(\pi_{k,t})$ and $\mu_4(\pi_{k,t})$ be the third and the fourth central moment of investor $k$'s share satisfying \eqref{eq:generalthird}, \eqref{eq:generalfourth} respectively.
If $\pi_{k,0} = \mathcal{O}(1/N^2)$, there exist $C_3, C_4 > 0$ independent of $t$ and $N$ such that
\begin{equation}
\label{eq:bound342}
\mu_3(\pi_{k,t}) \le \frac{C_3 \pi_{k,0}}{N^3},  \quad \mu_4(\pi_{k,t}) \le \frac{C_4 \pi_{k,0}}{N^4} \quad \mbox{for each } t \ge 1.
\end{equation}
\end{enumerate}
\end{lemma}

\begin{proof}
(1) By \cite[Lemma A.3]{RS21}, the sequence $(a_t, \, t \ge1)$ is increasing. 
So the lower bound in \eqref{eq:boundat2} follows from the fact that $a_t \ge a_1$.
Further by \eqref{eq:at}, we get
\begin{equation}
\label{eq:atupper}
a_t \le \sum_{n = 1}^t \left(\frac{R_n}{N + \sum_{k=1}^n R_k} \right)^2 \le \sum_{n = 1}^t \frac{R_n^2}{N^2}
\le \frac{\sum_{t \ge 1}R_t^2}{N^2},
\end{equation}
where the last inequality is due to the fact that $R_t = \Theta(t^{-\alpha})$ for $\alpha > \frac{1}{2}$ so $\sum_{t \ge 1}R_t^2 < \infty$.

(2) Note that $\sum_{n = 1}^t \left(\frac{R_n}{N + \sum_{k=1}^n R_k} \right)^k \le \sum_{t \ge 1} R_t^k/N^k$ for $k = 2, 3, \ldots$
Using the same argument as in Lemma \ref{lem:asympnozero} and the fact that $\pi_{k,0} = \mathcal{O}(1/N^2)$, we get
\begin{align*}
& \mu_3(\pi_{k,t}) \le \frac{C \pi_{k,0}}{N^2} \frac{1}{N^2} + C' \pi_{k,0} \frac{1}{N^3}, \\
& \mu_4(\pi_{k,t}) \le \frac{C \pi_{k,0}}{N^3} \frac{1}{N^2} +  \frac{C' \pi_{k,0}}{N^2} \frac{1}{N^3} + C'' \pi_{k,0} \frac{1}{N^4},
\end{align*}
for some $C, C', C'' > 0$ independent of $t$ and $N$.
This clearly yields the bounds \eqref{eq:bound342}.
\end{proof}

\begin{lemma}
\label{lem:smallerhalf}
Assume that the reward $R_t$ is decreasing, i.e. $R_t \ge R_{t+1}$ for each $t \ge 0$, and that $R_t = \Theta(t^{-\alpha})$ for $\alpha < \frac{1}{2}$.
\begin{enumerate}[itemsep = 3 pt]
\item
Let $a_t$ be defined by \eqref{eq:at}.  There exist $C > c > 0$ independent of $t$ and $N$ such that 
\begin{equation}
\label{eq:boundat3}
cN^{-\frac{1}{1-\alpha}} \le a_t \le CN^{-\frac{1}{1-\alpha}}, \quad \mbox{for each } t \ge N^{\frac{1}{1-\alpha}} .
\end{equation}
\item
Let $\mu_3(\pi_{k,t})$ and $\mu_4(\pi_{k,t})$ be the third and the fourth central moment of investor $k$'s share satisfying \eqref{eq:generalthird}, \eqref{eq:generalfourth} respectively.
If $\pi_{k,0} = \mathcal{O}(N^{-\frac{1}{1- \alpha}})$, there exist $C_3, C_4 > 0$ independent of $t$ and $N$ such that
\begin{equation}
\label{eq:bound343}
\mu_3(\pi_{k,t}) \le C_3 \pi_{k,0}N^{-\frac{2}{1 - \alpha}},  \quad \mu_4(\pi_{k,t}) \le C_4 \pi_{k,0}N^{-\frac{3}{1 - \alpha}} \quad \mbox{for each } t \ge 1.
\end{equation}
\end{enumerate}
\end{lemma}

\begin{proof}
(1) It follows from \eqref{eq:atlow} and \eqref{eq:atupper} that 
\begin{equation*}
\frac{a_t}{\sum_{n = 1}^t \left(\frac{R_n}{N + \sum_{k=1}^n R_k} \right)^2} \to 1 \quad \mbox{as }N \to \infty.
\end{equation*}
Thus, we need to study the asymptotic behavior of $\sum_{n = 1}^t \left(\frac{R_n}{N + \sum_{k=1}^n R_k} \right)^2$ as $N \to \infty$.
Since $R_t = \Theta(t^{-\alpha})$ for $\alpha < \frac{1}{2}$, we have:
\begin{equation}
\label{eq:133}
c  \sum_{n = 1}^t \frac{1}{(N n^{\alpha} + n)^2}.\le \sum_{n = 1}^t \left(\frac{R_n}{N + \sum_{k=1}^n R_k} \right)^2 \le C \sum_{n = 1}^t \frac{1}{(N n^{\alpha} + n)^2}.
\end{equation}
for some $C > c > 0$ independent of $t$ and $N$.
Again by the sum-integral trick, we get
\begin{equation}
\label{eq:134}
\int_{1}^{t+1} \frac{ds}{(N s^{\alpha} + s)^2} \le \sum_{n = 1}^t \frac{1}{(N n^{\alpha} + n)^2} \le \frac{1}{N^2} + \int_1^{\infty}\frac{ds}{(Ns^{\alpha} + s)^2}.
\end{equation}
Now for $t \ge N^{\frac{1}{1-\alpha}}$, we have
\begin{align}
\label{eq:135}
\int_1^{t}\frac{ds}{(Ns^{\alpha} + s)^2} & = \int_1^{N^{\frac{1}{1-\alpha}}}\frac{ds}{(Ns^{\alpha} + s)^2} + \int_{N^{\frac{1}{1-\alpha}}}^{t}\frac{ds}{(Ns^{\alpha} + s)^2} \notag \\
& \le \frac{1}{N^2} \int_1^{N^{\frac{1}{1-\alpha}}} \frac{ds}{s^{2 \alpha}} + \int_{N^{\frac{1}{1-\alpha}}}^{t} \frac{ds}{s^2} \le C' N^{-\frac{1}{1-\alpha}},
\end{align}
and
\begin{align}
\label{eq:136}
\int_1^{t}\frac{ds}{(Ns^{\alpha} + s)^2} & = \int_1^{N^{\frac{1}{1-\alpha}}}\frac{ds}{(Ns^{\alpha} + s)^2} + \int_{N^{\frac{1}{1-\alpha}}}^{t}\frac{ds}{(Ns^{\alpha} + s)^2} \notag \\
& \ge \frac{1}{4 N^2} \int_1^{N^{\frac{1}{1-\alpha}}} \frac{ds}{s^{2 \alpha}} + \int_{N^{\frac{1}{1-\alpha}}}^{t} \frac{ds}{4 s^2} \ge C'' N^{-\frac{1}{1-\alpha}},
\end{align}
Combining \eqref{eq:133}, \eqref{eq:134}, \eqref{eq:135} and \eqref{eq:136} yields 
\begin{equation*}
cC' N^{-\frac{1}{1-\alpha}} \le \sum_{n = 1}^t \left(\frac{R_n}{N + \sum_{k=1}^n R_k} \right)^2 \le CC' N^{-\frac{1}{1-\alpha}}
\quad \mbox{for } t \ge N^{\frac{1}{1-\alpha}}.
\end{equation*}
This leads to the bounds \eqref{eq:boundat3}.

(2) Similar to \eqref{eq:133}--\eqref{eq:135}, we have $\sum_{n = 1}^t \left(\frac{R_n}{N + \sum_{k=1}^n R_k} \right)^{k} \le C N^{-\frac{k-1}{1 - \alpha}}$ for $k = 2, 3, \ldots$
Using the same argument as in Lemma \ref{lem:asympnozero} and the fact that $\pi_{k,0} = \mathcal{O}(N^{-\frac{1}{1-\alpha}})$, we get
\begin{align*}
& \mu_3(\pi_{k,t}) \le C \pi_{k,0} N^{-\frac{1}{1-\alpha}} N^{-\frac{1}{1-\alpha}}+ C' \pi_{k,0} N^{-\frac{2}{1-\alpha}}, \\
& \mu_4(\pi_{k,t}) \le C \pi_{k,0} N^{-\frac{2}{1-\alpha}} N^{-\frac{1}{1-\alpha}}+ C' \pi_{k,0} N^{-\frac{1}{1-\alpha}} N^{-\frac{2}{1-\alpha}} + C''\pi_{k,0} N^{-\frac{3}{1-\alpha}},
\end{align*}
for some $C, C', C'' > 0$ independent of $t$ and $N$.
This gives the bounds \eqref{eq:bound343}.
\end{proof}

\begin{proof}[Proof of Theorem \ref{thm:2}]
(1) ($i$) By Chebyshev' inequality and the upper bound in \eqref{eq:boundat}, we get
\begin{equation*}
\mathbb{P}\left(\left|\frac{\pi_{k, t}}{\pi_{k,0}} - 1\right| > \varepsilon \right) \le \frac{a_t(1-\pi_{k,0})}{\pi_{k,0} \varepsilon^2}
\le \frac{R_1}{N \pi_{k, 0} \varepsilon^2} = \frac{R_1}{n_{k,0} \varepsilon^2}.
\end{equation*}
This proves \eqref{eq:decinter1}: the ratio $\pi_{k,t}/\pi_{k,0}$ converges in probability to $1$ as $n_{k,0} = f(N) \to \infty$.

($ii$) Note that $\var \left(\frac{\pi_{k, \infty}}{\pi_{k,0}}\right) = \frac{N a_{\infty} (1 - \pi_{k,0})}{n_{k,0}}$.
It follows from \eqref{eq:boundat} that
\begin{equation}
\label{eq:Nainfbound}
\frac{\underline{R}^2}{R_1} \le \liminf_{N \to \infty} N a_{\infty} \le \limsup_{N \to \infty} N a_{\infty} \le R_1.
\end{equation}
Since $n_{k,0} = \Theta(1)$, we deduce from the lower bound of \eqref{eq:Nainfbound} that there exists $c > 0$ independent of $N$ such that $\var \left(\frac{\pi_{k, \infty}}{\pi_{k,0}}\right) \ge c$.
For $\varepsilon > 0$, we have
\begin{align}
\label{eq:128}
c \le \var \left(\frac{\pi_{k, \infty}}{\pi_{k,0}}\right) & = \mathbb{E}\left(\left(\frac{\pi_{k, \infty}}{\pi_{k,0}} - 1\right)^2 1_{\left\{\left|\frac{\pi_{k, \infty}}{\pi_{k,0}} - 1 \right| > \varepsilon\right\}} \right) + \mathbb{E}\left(\left(\frac{\pi_{k, \infty}}{\pi_{k,0}} - 1\right)^2 1_{\left\{\left|\frac{\pi_{k, \infty}}{\pi_{k,0}} - 1 \right| \le \varepsilon \right\}} \right) \notag \\
& \le \sqrt{\mu_4 \left(\frac{\pi_{k, \infty}}{\pi_{k,0}} \right)  \mathbb{P}\left(\left|\frac{\pi_{k, \infty}}{\pi_{k,0}} - 1 \right| > \varepsilon\right)}  + \varepsilon^2.
\end{align}
By the upper bound \eqref{eq:bound34}, we get $\mu_4 \left(\frac{\pi_{k, \infty}}{\pi_{k,0}} \right) \le \frac{C_4}{n_{k,0}^3}: = C$.
As a result, for $\varepsilon < \sqrt{c}$, 
\begin{equation*}
 \mathbb{P}\left(\left|\frac{\pi_{k, \infty}}{\pi_{k,0}} - 1 \right| > \varepsilon\right) \ge \frac{(c - \varepsilon^2)^2}{C},
\end{equation*}
which leads to the anti-concentration bound \eqref{eq:decsmall1}.

($iii$) Again from \eqref{eq:Nainfbound} we obtain
\begin{equation*}
\frac{\underline{R}^2}{R_1} \le \liminf_{N \to \infty} n_{k, 0} \var \left(\frac{\pi_{k, \infty}}{\pi_{k,0}}\right) \le \limsup_{N \to \infty} n_{k, 0} \var \left(\frac{\pi_{k, \infty}}{\pi_{k,0}}\right) \le R_1,
\end{equation*}
which implies $\var \left(\frac{\pi_{k, \infty}}{\pi_{k,0}}\right) \to \infty$ since $n_{k,0} = o(1)$.

(2) and (3) follow from Lemma \ref{lem:largerhalf} and \ref{lem:smallerhalf} exactly in the same way as (1) is derived from Lemma \ref{lem:asympnozero}.
\end{proof}

\section{Proof of Theorem \ref{thm:3}}
\label{scC}

\quad To prove Theorem \ref{thm:3}, we need the following lemma.
\begin{lemma}
\label{lem:polyat}
Assume that the coin reward $R_t = \rho N_{t-1}^{\gamma}$ for some $\rho > 0$ and $\gamma \in [0,1)$.
\begin{enumerate}[itemsep = 3 pt]
\item
Let $a_t$ be defined by \eqref{eq:at}.  There exist $c > 0$ independent of $t$ and $N$ such that 
\begin{equation}
\label{eq:boundat4}
cN^{\gamma-1} \le a_t \le \frac{\rho}{1- \gamma}N^{\gamma-1}, \quad \mbox{for } t \mbox{ sufficiently large} .
\end{equation}
\item
Let $\mu_3(\pi_{k,t})$ and $\mu_4(\pi_{k,t})$ be the third and the fourth central moment of investor $k$'s share satisfying \eqref{eq:generalthird}, \eqref{eq:generalfourth} respectively.
If $\pi_{k,0} = \mathcal{O}(N^{\gamma-1})$, there exist $C_3, C_4 > 0$ independent of $t$ and $N$ such that
\begin{equation}
\label{eq:bound344}
\mu_3(\pi_{k,t}) \le C_3 \pi_{k,0}N^{2\gamma -2},  \quad \mu_4(\pi_{k,t}) \le C_4 \pi_{k,0}N^{3\gamma-3} \quad \mbox{for each } t \ge 1.
\end{equation}
\end{enumerate}
\end{lemma}

\begin{proof}
(1) The upper bound $a_t \le \frac{\rho}{1-\gamma} N^{\gamma -1}$ follows from \cite[Lemma A.5]{RS21}.
Recall that $N_{t+1} = N_t + \rho N_t^{\gamma}$, and thus $N_t$ behaves asymptotically as $t^{\frac{1}{1-\gamma}}$.
As a result, there is $c> 0$ independent of $t$ and $N$ such that $N_{t} \ge c N_{t+1}$ for $t$ sufficiently large.
This also implies that $R_{t} \ge c^{\gamma} R_{t+1}$ for $t$ sufficiently large.
By \eqref{eq:atlow}, we have
\begin{align}
\label{eq:145}
a_t & \ge \left(1 - \rho N^{\gamma-1}\right) \sum_{n = 1}^t \left(\frac{R_n}{N_n} \right)^2 \notag\\
& \ge c' \sum_{n = 1}^t \left(\frac{R_{n+1}}{N_n} \right)^2 =  c' \sum_{n = 1}^t \frac{N_{n+1} - N_n}{N_n^{2 -\gamma}},
\end{align}
where the last equality is due to the fact that $R_{n+1}^2 =  \rho N_n^{\gamma}(N_{n+1} - N_n)$.
Using the sum-integral trick, we get
\begin{equation}
\label{eq:146}
\sum_{n = 1}^t \frac{N_{n+1} - N_n}{N_n^{2 -\gamma}} \ge c'' \int_N^{N_t} \frac{ds}{s^{2 - \gamma}} \ge c''' N^{\gamma -1},
\end{equation}
for $t$ sufficiently large.
Combining \eqref{eq:145} and \eqref{eq:146} yields the lower bound in \eqref{eq:boundat4}.

(2) Similar to the proof of the upper bound in \eqref{eq:boundat4}, we can show that
$\sum_{n = 1}^t \left(\frac{R_n}{N_n} \right)^k \le \rho^{k-1} N^{(k-1)(\gamma - 1)}$ for $k = 2, 3,\ldots$
Using the same argument as in Lemma \ref{lem:asympnozero} and the fact that $\pi_{k,0} = \mathcal{O}(N^{\gamma - 1})$,
we obtain:
\begin{align*}
& \mu_3(\pi_{k,t}) \le C \pi_{k,0} N^{\gamma - 1} N^{\gamma -1}+ C' \pi_{k,0} N^{2(\gamma - 1)}, \\
& \mu_4(\pi_{k,t}) \le C \pi_{k,0} N^{2(\gamma - 1)} N^{\gamma -1}+ C' \pi_{k,0} N^{\gamma -1} N^{2(\gamma -1)} + C''\pi_{k,0} N^{3(\gamma -1)},
\end{align*}
for $C, C', C'' > 0$ independent of $t$ and $N$.
 This leads to the bounds \eqref{eq:bound344}.
\end{proof}

\begin{proof}[Proof of Theorem \ref{thm:3}]
(1) It follows from the proof of Proposition 4 of \cite{RS21} that for $R_t = \rho N_{t-1}^{\gamma}$ with $\gamma > 1$, 
the sequence $a_t$ increases to the limit $a_{\infty} = 1$.
Consequently, 
\begin{equation}
\mathbb{E}(\pi_{k,t}^2) \to \mathbb{E}(\pi_{k, \infty}^2) = \pi_{k,0} \quad \mbox{as } t \to \infty.
\end{equation}
Note that $\mathbb{E}(\pi_{k,\infty}) = \pi_{k,0}$, so $\mathbb{E}(\pi_{k, \infty} (1 - \pi_{k, \infty})) = 0$.
Since $\pi_{k, \infty} (1 - \pi_{k, \infty}) \ge 0$, we get 
$\pi_{k,\infty} \in \{0,1\}$ and \eqref{eq:incextreme} holds.

(2) follows from Lemma \ref{lem:polyat} exactly in the same way as Theorem \ref{thm:1} follows from Lemmas \ref{lem:asympnozero}--\ref{lem:smallerhalf}.
\end{proof}

\section{Proof of Theorems \ref{thm:4} -- \ref{thm:6}}
\label{scD}

\subsection{Proof of Theorem \ref{thm:4}}
\label{scD1}

To prove Theorem \ref{thm:4}, we need the following result of the Blackwell-MacQueen urn scheme which generalizes the P\'olya urn scheme.

\begin{lemma}[\cite{BM73}]
\label{lem:BMscheme}
Let $\mu$ be a positive and finite measure on a Polish space $(S, \mathcal{S})$.
Define a sequence $(X_t, \, t = 1,2, \ldots)$ as follows: 
$X_1$ is distributed as $\mu(\cdot)/\mu(S)$, and for $t \ge 1$,
\begin{equation}
\label{eq:prediction}
\mathbb{P}(X_{t+1} \in \cdot | X_1, \ldots, X_t) = \frac{\mu(\cdot) + \sum_{n = 1}^t \delta_{X_n}(\cdot)}{\mu(S) + t}:= F_t(\cdot),
\end{equation}
where $\delta_X(\cdot)$ is the Dirac mass at point $X$.
Then
\begin{itemize}[itemsep = 3 pt]
\item
$F_t$ converges in total variation (and thus in distribution) almost surely to a random discrete distribution $F$, which has $\Dir(\mu)$ distribution.
 \item
Conditional given $F$, $X_1, X_2, \ldots$ are independent and identically distributed as $F$.
 \end{itemize}
\end{lemma}

\begin{proof}[Proof of Theorem \ref{thm:4}]
First assume that $R_t \equiv R$, and let $\mu$ be a positive measure on $S = \mathbb{N}$ such that $\mu(\{k\}) = \frac{n_{k,0}}{R}$, 
so $\mu(S) = \frac{N}{R}$.
Let $X_t, \, t \ge 1$ be the index of the investor who is selected by the PoS protocol at time $t$.
By definition of the PoS scheme, we have
\begin{align*}
& \mathbb{P}(X_1 = k) = \frac{n_{k,0}}{N} = \frac{\mu(\{k\})}{\mu(S)}, \\
& \mathbb{P}(X_{t+1}= k | X_1, \ldots, X_t) = \frac{n_{k,0} + R \sum_{n = 1}^t 1_{\{X_n = k\}}}{N + Rt} = \frac{\mu(\{k\}) + \sum_{n = 1}^t \delta_{X_n}(\{k\})}{\mu(S) + t}.
\end{align*}
Lemma \ref{lem:BMscheme} then implies the identity in distribution \eqref{eq:limDirinf}.
By Definition \ref{def:Dirichlet}, we get $\pi_{k,0} \stackrel{d}{=} \bet(\frac{n_{k,0}}{R}, \frac{N-n_{k,0}}{R})$ for each fixed $k$, and thus
the results in Theorem \ref{thm:1} hold.
The results in Theorems \ref{thm:2} and \ref{thm:3} are stated for each investor $k$, and do not depend on the number $K$ of investors.
So these results also hold in the infinite population setting. 
\end{proof}

\subsection{Proof of Theorem \ref{thm:5}}
\label{scD2}

\begin{proof}
The stick-breaking representation \eqref{eq:stickbreaking} follows from the Blackwell-MacQueen urn construction (Lemma \ref{lem:BMscheme}), along with various constructions of the Dirichlet measure by \cite{Fer73, Mc65}.
See e.g. \cite[Section 2.2]{Pitman96} for a review of the circle of ideas.
The fact that $K_t$ behaves asymptotically as $\frac{N}{R} \log t$ is read from \cite[Theorem 2.3 ]{KH73}.
\end{proof}

\subsection{Proof of Theorem \ref{thm:6}}
\label{scD3}
\begin{proof}
Note that $n_{k,t+1} = n_{k,t} + R_t$ with probability $\frac{n_{k,t}}{N_t + \theta}$ 
and $n_{k,t+1} = n_{k,t}$ with probability $1-\frac{n_{k,t}}{N_t + \theta}$.
As a result,
\begin{align*}
\mathbb{E}(\pi_{k,t+1}|\mathcal{F}_t) &= \frac{n_{k,t} + R_t}{N_t + R_t} \frac{n_{k,t}}{N_t + \theta} + \frac{n_{k,t}}{N_t + R_t} 
\left(1-\frac{n_{k,t}}{N_t + \theta} \right) \\
& = \frac{n_{k,t}(N_t + \theta + R_t)}{(N_t + R_t)(N_t + \theta)} < \frac{n_{k,t}}{N_t} = \pi_{k,t}.
\end{align*}
So $(\pi_{k,t}, \, t \ge 0)$ is a supermartingale. 
By the martingale convergence theorem, $(\pi_{1,t}, \ldots, \pi_{K,t})$ converges almost surely to a random vector 
$(\pi_{1,\infty}, \ldots, \pi_{K,\infty})$.
Observe that $\mathbb{E}(\pi_{k,t+1}) = \frac{N_t(N_{t+1} + \theta)}{N_{t+1}(N_t + \theta)} \mathbb{E}(\pi_{k,t})$ which implies that
\begin{equation}
\label{eq:infprod}
\mathbb{E}(\pi_{k,\infty}) = \pi_{k,0} \prod_{t = 1}^{\infty}\left(1 -  \frac{\theta R_{t+1}}{ N_{t+1}(N_t + \theta)}\right).
\end{equation}
(1) Assume that $R_t$ is decreasing.
If $\lim_{t \to \infty} R_t = R > 0$, we have $\frac{R_{t+1}}{ N_{t+1}(N_t + \theta)} = \mathcal{O}(t^{-2})$, and 
if $R_t = \Theta(t^{-\alpha})$ for $\alpha < 1$, we have $\frac{R_{t+1}}{ N_{t+1}(N_t + \theta)} = \mathcal{O}(t^{-2 + \alpha})$.
In both cases, we get $\sum_{t = 1}^{\infty} \frac{R_{t+1}}{ N_{t+1}(N_t + \theta)} < \infty$ and thus the infinite product in \eqref{eq:infprod} converges to some number in $(0,1)$.
If $R_t =\Theta(t^{-\alpha})$ for $\alpha > 1$, we have $\frac{R_{t+1}}{ N_{t+1}(N_t + \theta)} \ge C t^{-1}$ for $t$ sufficiently large.
In this case, $\sum_{t = 1}^{\infty} \frac{R_{t+1}}{ N_{t+1}(N_t + \theta)} = \infty$.
Consequently, we get $\mathbb{E}(\pi_{k,\infty}) = 0$ which implies that $\pi_{k,\infty} = 0$ almost surely.

(2) Assume that $R_t = \rho N_{t-1}^{\gamma}$ for $\rho, \gamma > 0$.
If $\gamma < 1$, it follows from the proof of Lemma \ref{lem:polyat} that $R_t = \Theta\left(t^{\frac{\gamma}{1 - \gamma}}\right)$.
We have $\frac{R_{t+1}}{ N_{t+1}(N_t + \theta)} = \mathcal{O}\left(t^{-2 - \frac{\gamma}{1 - \gamma}}\right)$.
If $\gamma > 1$, then $\frac{R_{t+1}}{ N_{t+1}(N_t + \theta)} = \Theta(N_t^{-1})$ and $N_t$ grows exponentially in $t$.
In both cases, we have $\sum_{t = 1}^{\infty} \frac{R_{t+1}}{ N_{t+1}(N_t + \theta)} < \infty$ 
which implies that the infinite product in \eqref{eq:infprod} converges to some number in $(0,1)$. 

(3) It is easily checked that the PoS scheme \eqref{eq:selection1f}--\eqref{eq:selectiontf} with a constant reward is just the Blackwell-MacQueen urn with $\mu:= (\sum_{k = 1}^K n_{k,0} \delta_k + \theta \nu)/R$.
It follows from Lemma \ref{lem:BMscheme} that selection probability \eqref{eq:selectiontf} converges almost surely to a random discrete distribution $F \stackrel{d}{=} Dir(\mu)$, 
and given $F$ the indices of investors selected are independent and identically distributed as $F$.
This implies that the limiting share of investor $k$ is $\bet\left(\frac{n_{k,0}}{R}, \frac{N + \theta - n_{k,0}}{R}\right)$, and the results in Theorem \ref{thm:1} follow.
\end{proof}

\section{Numerical illustrations for Theorems \ref{thm:2} and \ref{thm:3}}
\label{scE}

\subsection{Numerical illustrations for Theorem \ref{thm:2}: decreasing reward} \hfill\par
\label{scE1}

\smallskip
(1) $R_t$ is bounded from $0$:
Figure \ref{fig:3} shows the concentration bound \eqref{eq:decinter1} for large investors;
Figure \ref{fig:4A} shows the bounded variance and the anti-concentration bound \eqref{eq:decsmall1} for medium investors;
Figure \ref{fig:4B} shows the exploding variance for small investors.

\begin{figure}[htb]
 \includegraphics[width= 0.45\textwidth]{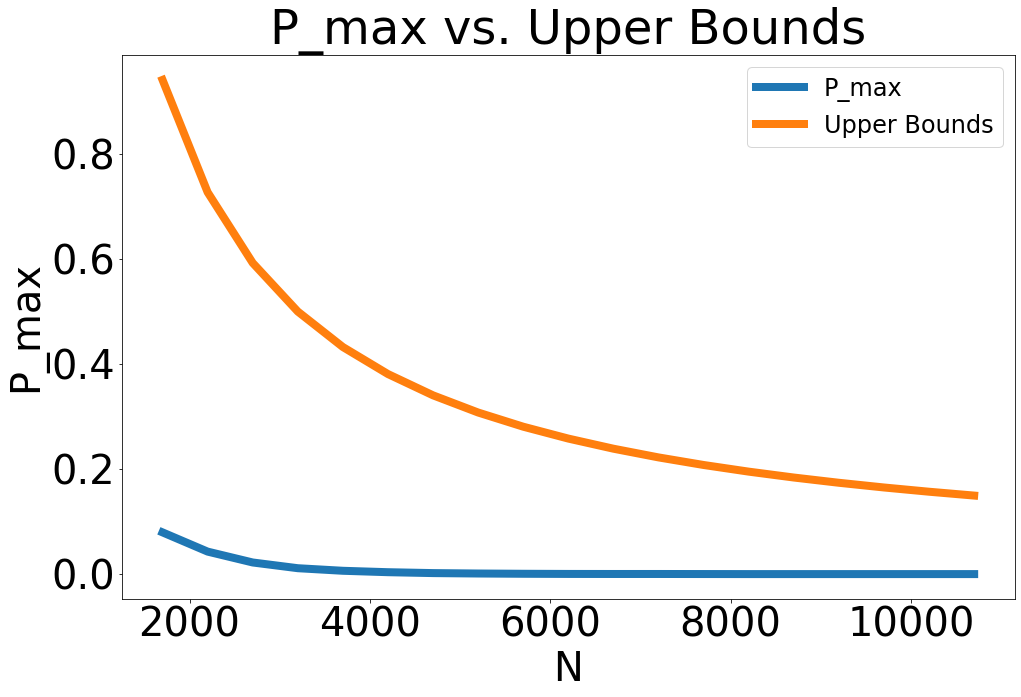}
\caption{Decreasing reward: stability of $\pi_{k,t}/\pi_{k,0}$ for large investors.
Blue curve: $P_{\max}$ is a MC estimate of $\max_{1 \le t \le 50000} \mathbb{P}\left(\left| \frac{\pi_{k,t}}{\pi_{k,0}}\right| > 0.05 \right)$.
Orange curve: right side upper bound in \eqref{eq:decinter1} with $R_t = 1 + 0.999^t$, $\varepsilon = 0.05$, $n_{k,0} = N/2$ and $N \in \{1700, 2200, 2700, ..., 10700\}$.}
\label{fig:3}
\end{figure}

\begin{figure}[htb]
    \centering
\begin{subfigure}{0.45\textwidth}
  \includegraphics[width=\linewidth]{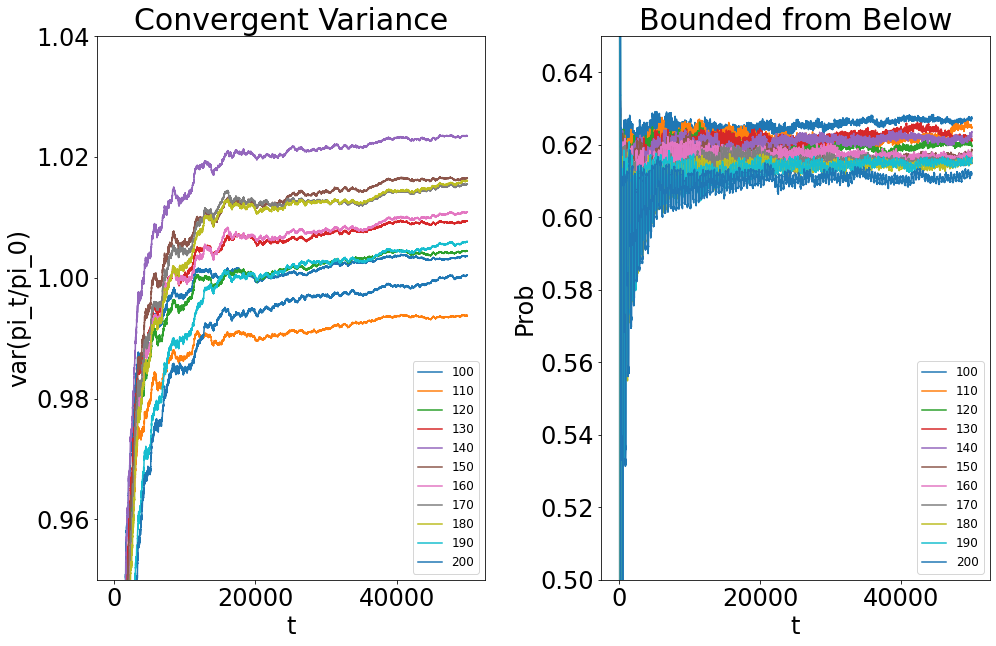}
  \caption{MC estimates of  $\var(\pi_{k,t}/\pi_{k,0})$ and $\mathbb{P}\left( \left|\frac{\pi_{k,t}}{\pi_{k,0}} - 1 \right| >   0.5\right)$
  with $R_t = 1 + t^{-1}$, $n_{k,0} =1$ and $N \in \{100, 110, 120, ..., 200\}$.}
    \label{fig:4A}
\end{subfigure}\hfil
\begin{subfigure}{0.45\textwidth}
  \includegraphics[width=\linewidth]{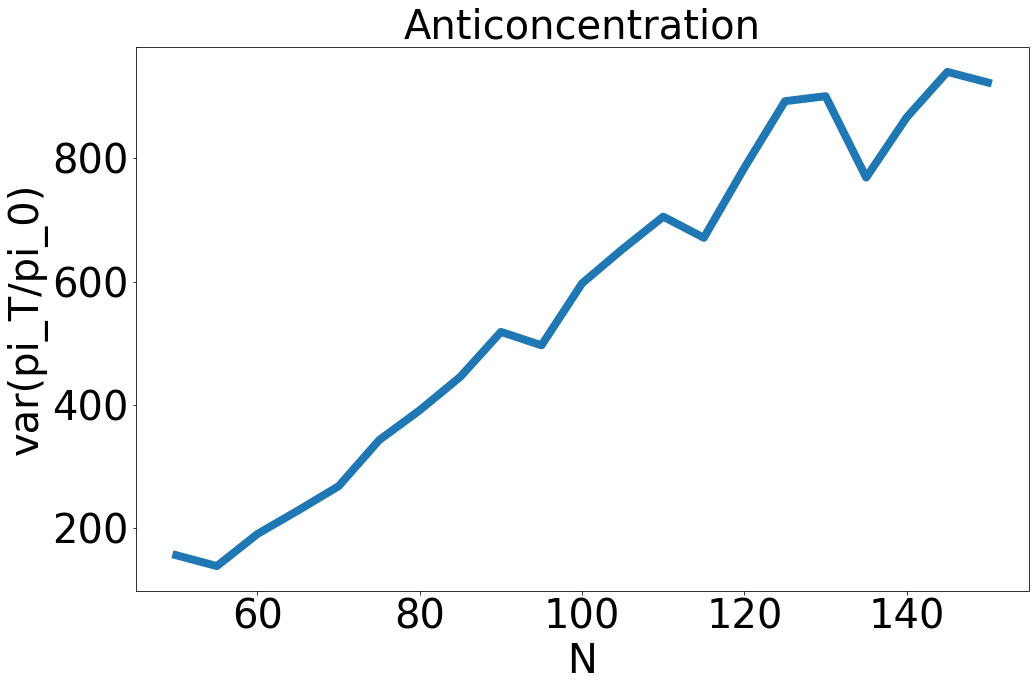}
  \caption{MC estimate of $\var(\pi_{k,50000}/\pi_{k,0})$ with $R_t = 1 + t^{-1}$, $n_{k,0} = N^{-1.1}$ and $N \in \{100, 110, 120, ..., 300\}$.}
  \label{fig:4B}
\end{subfigure}\hfil
\caption{Decreasing reward: instability of $\pi_{k,t}/\pi_{k,0}$ for medium and small investors.}
\end{figure}

(2) $R_t = \Theta(t^{-\alpha})$ for $\alpha > 1/2$:
Figure \ref{fig:5} shows the concentration bound \eqref{eq:decinter2} for large investors;
Figure \ref{fig:6A} shows the bounded variance for medium investors;
Figure \ref{fig:6B} shows the exploding variance for small investors.

\begin{figure}[htb]
  \includegraphics[width=0.45 \textwidth]{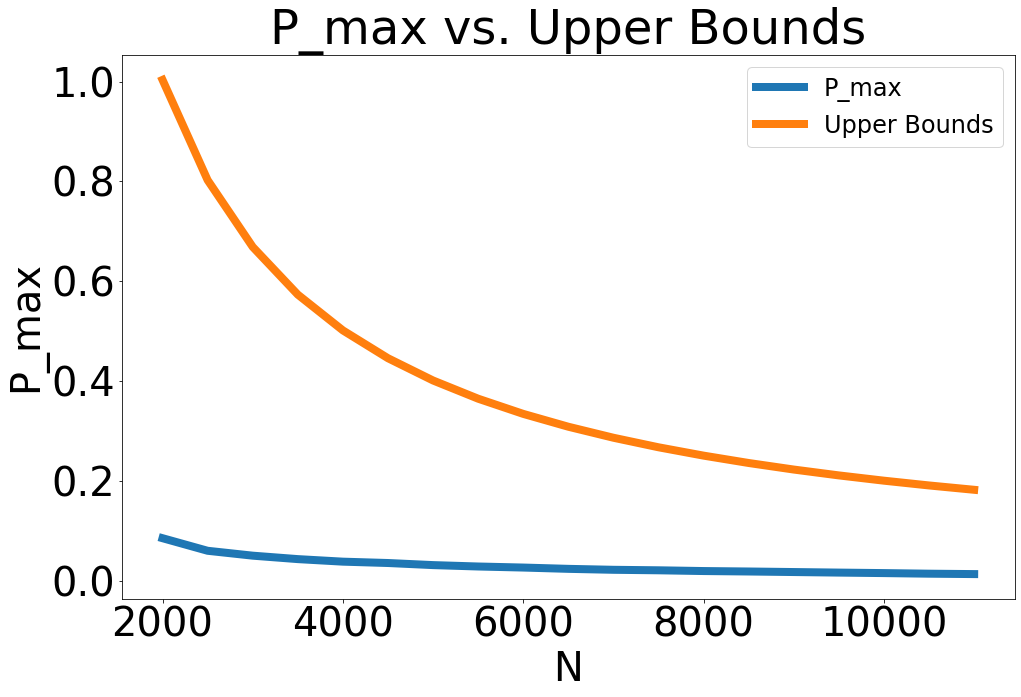}
\caption{Decreasing reward: stability of $\pi_{k,t}/\pi_{k,0}$ for large investors.
Blue curve: $P_{\max}$ is a MC estimate of $\max_{1 \le t \le 50000} \mathbb{P}\left(\left| \frac{\pi_{k,t}}{\pi_{k,0}}\right| > 0.05 \right)$.
Orange curve: right side upper bound in \eqref{eq:decinter2} with $R_t = t^{-0.6}$, $\varepsilon = 0.05$, $n_{k,0} = 1$ and $N \in \{2000, 2500, 3000, ..., 11000\}$.}
\label{fig:5}
\end{figure}

\begin{figure}[htb]
    \centering
\begin{subfigure}{0.42\textwidth}
  \includegraphics[width=\linewidth]{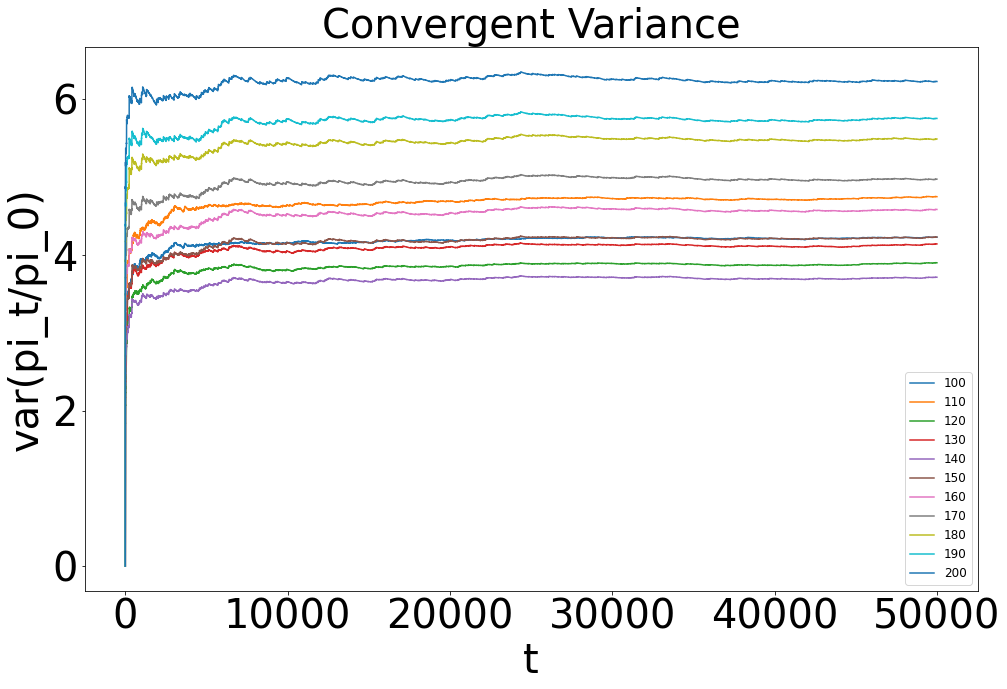}
  \caption{MC estimate of $\var(\pi_{k,t}/\pi_{k,0})$ with $R_t = t^{-0.6}$, $n_{k,0} = 1/N$ and $N \in \{100, 110, 120, ..., 200\}$.}
   \label{fig:6A}
\end{subfigure}\hfil
\begin{subfigure}{0.45\textwidth}
  \includegraphics[width=\linewidth]{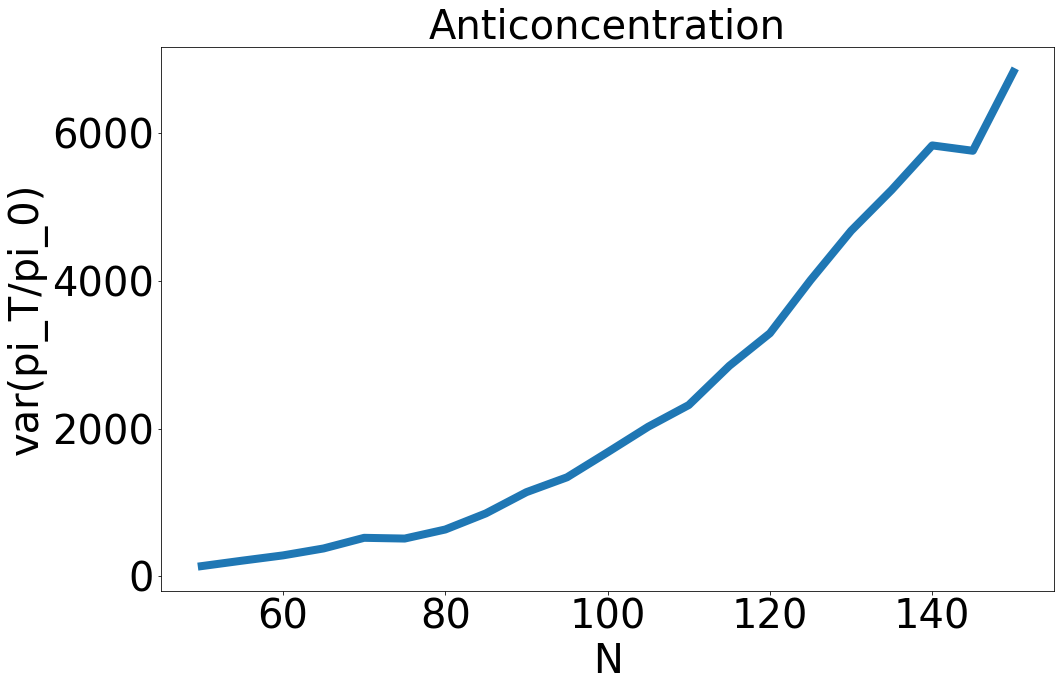}
  \caption{MC estimate of $\var(\pi_{k,50000}/\pi_{k,0})$ with $R_t = t^{-0.6}$, $n_{k,0} = 1/N^2$ and $N \in \{50, 55, 60, ..., 150\}$.}
  \label{fig:6B}
\end{subfigure}\hfil
\caption{Decreasing reward: instability of $\pi_{k,t}/\pi_{k,0}$ for medium and small investors.}
\end{figure}

(3) $R_t = \Theta(t^{-\alpha})$ for $\alpha < 1/2$:
Figure \ref{fig:7} shows the concentration bound \eqref{eq:decinter3} for large investors;
Figure \ref{fig:8A} shows the bounded variance and the anti-concentration bound \eqref{eq:decsmall3} for medium investors;
Figure \ref{fig:8B} shows the exploding variance for small investors.

\begin{figure}[htb]
  \includegraphics[width=0.45 \textwidth]{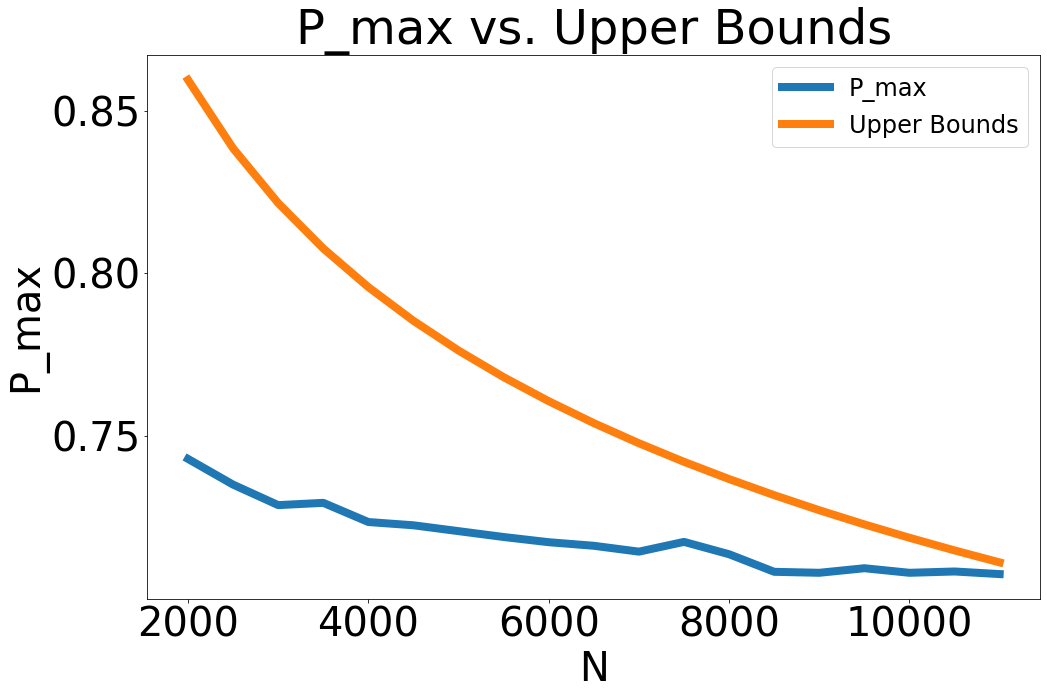}
\caption{Decreasing reward: stability of $\pi_{k,t}/\pi_{k,0}$ for large investors.
Blue curve: $P_{\max}$ is a MC estimate of $\max_{1 \le t \le 50000} \mathbb{P}\left(\left| \frac{\pi_{k,t}}{\pi_{k,0}}\right| > 0.25 \right)$.
Orange curve: right side upper bound in \eqref{eq:decinter3} with $R_t = t^{-0.1}$, $\varepsilon = 0.25$, $n_{k,0} = 1$ and $N \in \{2000, 2500, 3000, ..., 11000\}$.}
\label{fig:7}
\end{figure}

\begin{figure}[htb]
    \centering
\begin{subfigure}{0.45\textwidth}
  \includegraphics[width=\linewidth]{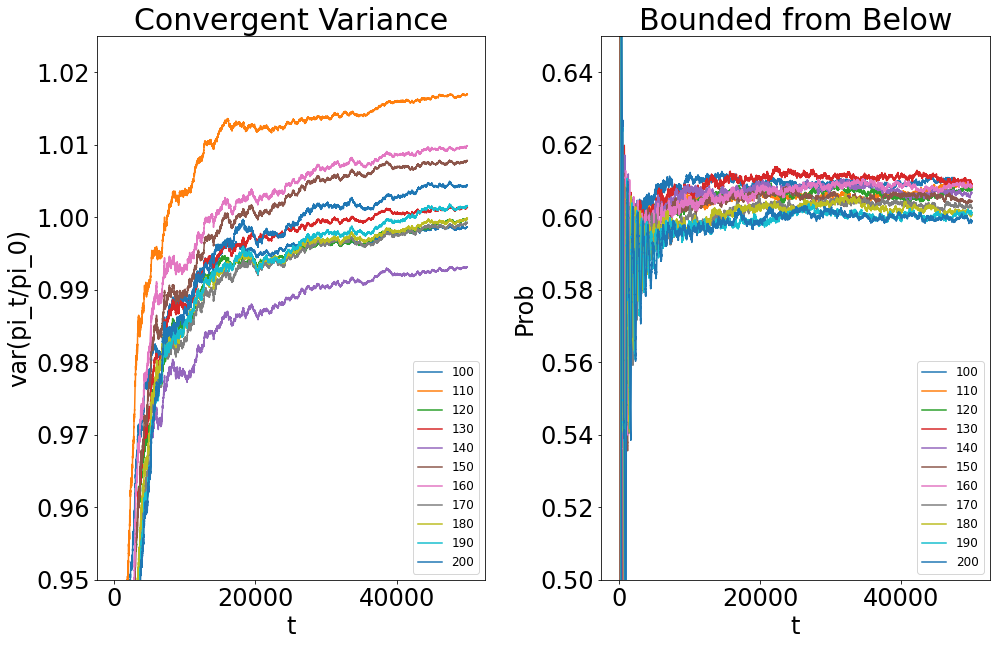}
  \caption{MC estimates of  $\var(\pi_{k,t}/\pi_{k,0})$ and $\mathbb{P}\left( \left|\frac{\pi_{k,t}}{\pi_{k,0}} - 1 \right| >   0.5\right)$
  with $R_t = t^{-0.1}$, $n_{k,0} =N^{-1/9}$ and $N \in \{100, 110, 120, ..., 200\}$.}
    \label{fig:8A}
\end{subfigure}\hfil
\begin{subfigure}{0.45\textwidth}
  \includegraphics[width=\linewidth]{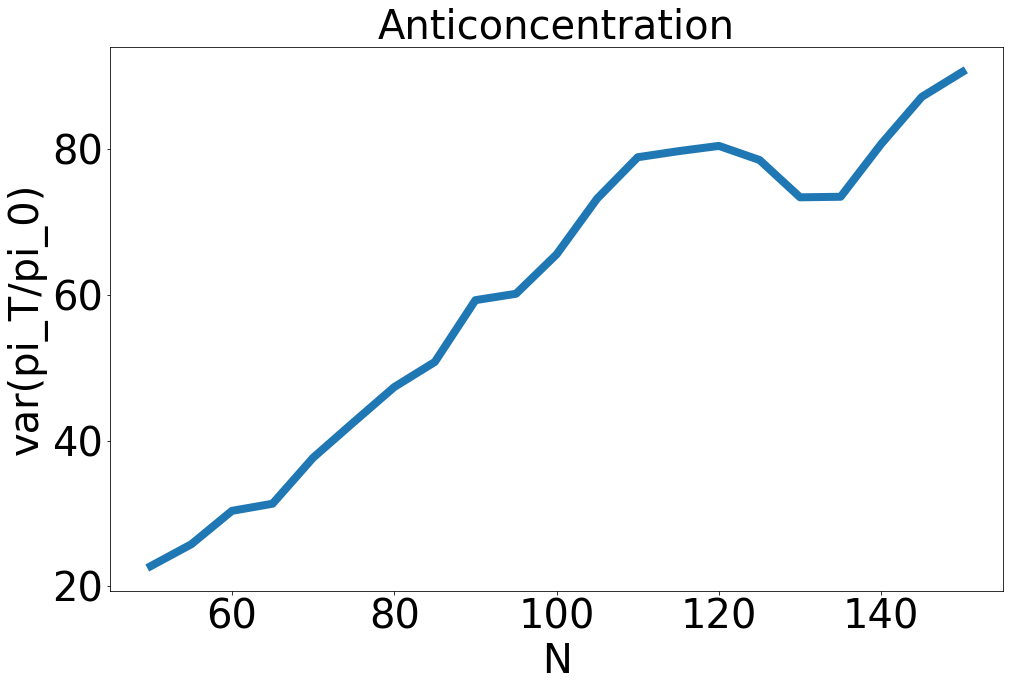}
  \caption{MC estimate of $\var(\pi_{k,50000}/\pi_{k,0})$ with $R_t = t^{-0.1}$, $n_{k,0} = 1/N$ and $N \in \{50, 55, 60, ..., 150\}$.}
  \label{fig:8B}
\end{subfigure}\hfil
\caption{Decreasing reward: instability of $\pi_{k,t}/\pi_{k,0}$ for medium and small investors.}
\end{figure}

\subsection{Numerical illustrations for Theorem \ref{thm:3}: increasing reward}
\label{scE2}
\hfill\par

(1) $R_t = \rho N_{t-1}^\gamma$ for $\gamma < 1$: 
Figure \ref{fig:9} shows the chaotic centralization under a geometric reward.

\begin{figure}[htb]
    \centering
\begin{subfigure}{0.3\textwidth}
  \includegraphics[width=\linewidth]{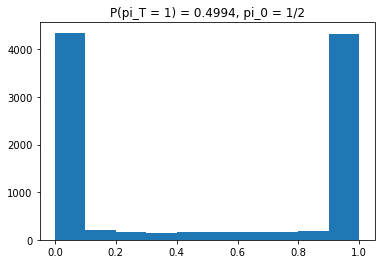}
  \caption{$\pi_{k,0} = \frac{1}{2}$}
\end{subfigure}\hfil
\begin{subfigure}{0.3\textwidth}
  \includegraphics[width=\linewidth]{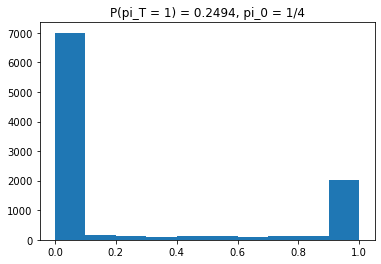}
  \caption{$\pi_{k,0} = \frac{1}{4}$}
\end{subfigure}\hfil
\begin{subfigure}{0.3\textwidth}
  \includegraphics[width=\linewidth]{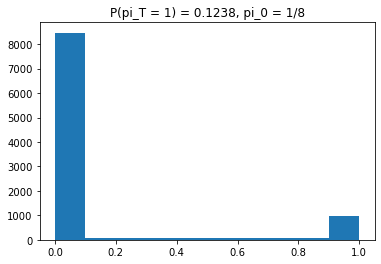}
  \caption{$\pi_{k,0} = \frac{1}{8}$}
\end{subfigure}
\caption{Increasing reward: histogram of $\pi_{k, 5000}$ with $\rho = 0.001$, $\gamma = 1.1$, $N = 1000$ and $\pi_{k,0} \in \{1/2, 1/4, 1/8\}$}
\label{fig:9}
\end{figure}

(2) $R_t = \rho N_{t-1}^\gamma$ for $\gamma < 1$: 
Figure \ref{fig:10} shows the concentration bound \eqref{eq:incinter} for large investors;
Figure \ref{fig:11A} shows the bounded variance and the anti-concentration bound \eqref{eq:incsmall} for medium investors;
Figure \ref{fig:11B} shows the exploding variance for small investors.

\begin{figure}[htb]
  \includegraphics[width=0.45\textwidth]{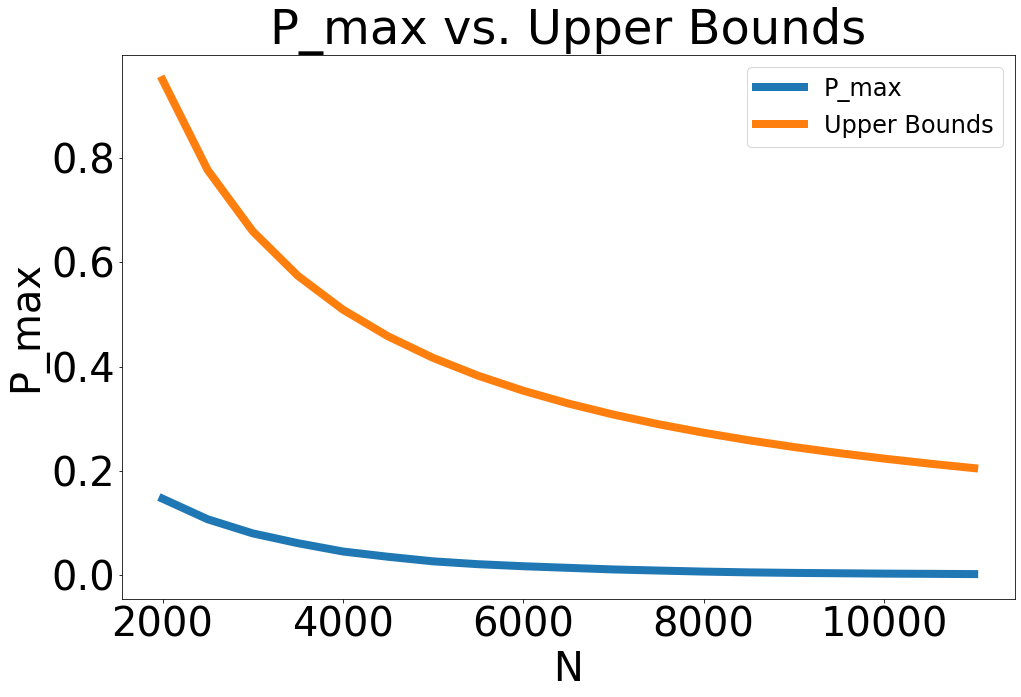}
\caption{Increasing reward: stability of $\pi_{k,t}/\pi_{k,0}$ for large investors.
Blue curve: $P_{\max}$ is a MC estimate of $\max_{1 \le t \le 50000} \mathbb{P}\left(\left| \frac{\pi_{k,t}}{\pi_{k,0}}\right| > 0.05 \right)$.
Orange curve: right side upper bound in \eqref{eq:decinter2} with $\rho = 1$, $\gamma = 0.1$, $\varepsilon = 0.05$, $n_{k,0} = N/2$ and $N \in \{2000, 2500, 3000, ..., 11000\}$.}
\label{fig:10}
\end{figure}

\begin{figure}[htb]
    \centering
\begin{subfigure}{0.45\textwidth}
  \includegraphics[width=\linewidth]{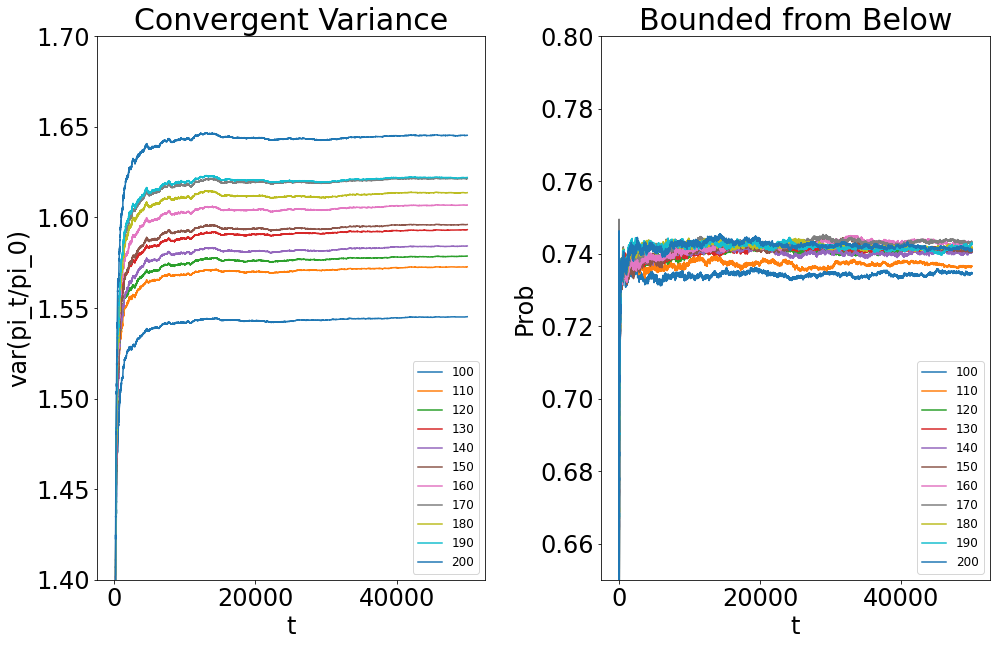}
  \caption{MC estimates of  $\var(\pi_{k,t}/\pi_{k,0})$ and $\mathbb{P}\left( \left|\frac{\pi_{k,t}}{\pi_{k,0}} - 1 \right| >   0.5\right)$
  with $\rho = 1$, $\gamma = 0.1$, $n_{k,0} =N^{1/2}$ and $N \in \{100, 110, 120, ..., 200\}$.}
  \label{fig:11A}
\end{subfigure}\hfil
\begin{subfigure}{0.45\textwidth}
  \includegraphics[width=\linewidth]{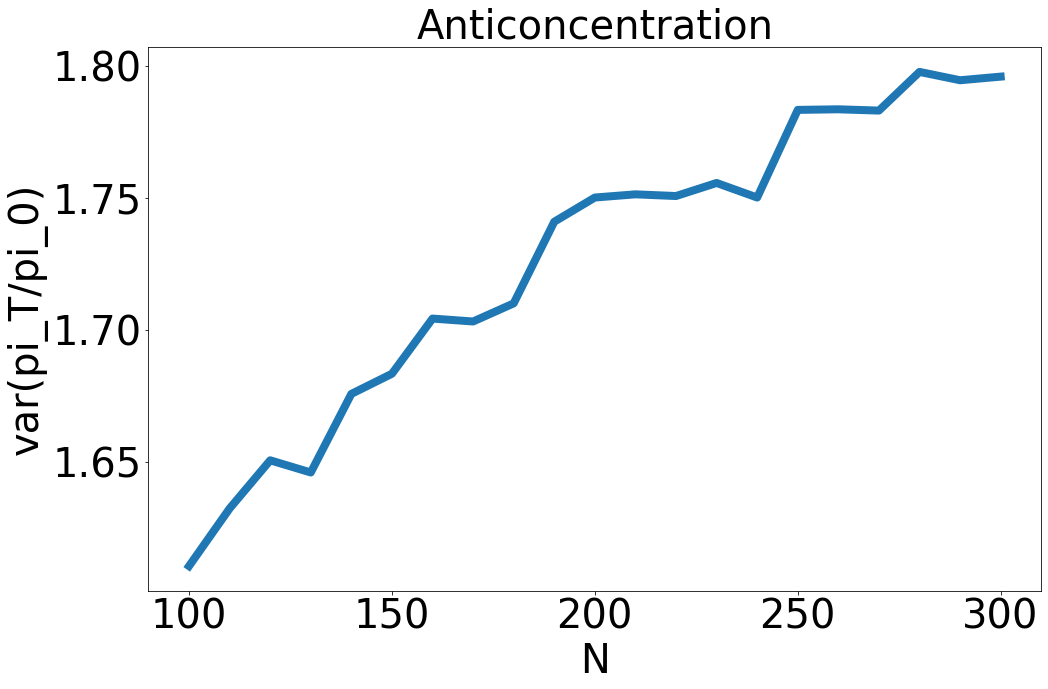}
  \caption{MC estimate of $\var(\pi_{k,50000}/\pi_{k,0})$ with $\rho = 1$, $\gamma = 0.1$, $n_{k,0} = N^{0.01}$ and $N \in \{100, 110, 120, ..., 300\}$.}
  \label{fig:11B}
\end{subfigure}\hfil
\caption{Increasing reward: case $(ii)$ and case $(iii)$ investors}
\end{figure}

\clearpage

\bibliographystyle{abbrvnat}
\bibliography{unique}
\end{document}